\documentclass[reqno,11pt,letterpaper]{amsart}

\usepackage{amsmath,amssymb,amsthm,mathrsfs,bbm,lipsum}
\usepackage{accents}
\usepackage{arydshln}
\usepackage{upgreek}
\usepackage{slashed}
\usepackage{xifthen}
\usepackage{graphicx}
\usepackage{longtable}
\usepackage[inline]{enumitem}
\usepackage[all]{xy}
\usepackage{braket}
\usepackage{subfig}
\usepackage{xcolor}
\usepackage{colortbl}

\usepackage{xcolor}
\definecolor{winered}{rgb}{0.6,0,0}
\definecolor{lessblue}{rgb}{0,0,0.7}

\usepackage[pdftex,colorlinks=true,linkcolor=winered,citecolor=lessblue,breaklinks=true,bookmarksopen=true]{hyperref}

\makeatletter
\newcommand{\myitem}[3]{\item[#2]\def\@currentlabel{#3}\label{#1}}
\makeatother

\usepackage{titletoc}

\makeatletter

\def\@tocline#1#2#3#4#5#6#7{
\begingroup
  \par
    \parindent\z@ \leftskip#3 \relax \advance\leftskip\@tempdima\relax
                  \rightskip\@pnumwidth plus 4em \parfillskip-\@pnumwidth
    \ifcase #1 
       \vskip 0.6em \hskip 0em 
       \or
       \or \hskip 0em 
       \or \hskip 1em 
    \fi%
    #6
    \nobreak\relax{\leavevmode\leaders\hbox{\,.}\hfill}
    \hbox to\@pnumwidth {\@tocpagenum{#7}}
  \par
\endgroup
}

 \def\l@section{\@tocline{0}{0pt}{0pc}{}{}}

\renewcommand{\tocsection}[3]{%
  \indentlabel{\@ifnotempty{#2}{ 
    \ignorespaces\bfseries{#2. #3}}}
  \indentlabel{\@ifempty{#2}{\ignorespaces\bfseries{#3}}{}} 
    \vspace{1.5pt}}

\renewcommand{\tocsubsection}[3]{%
  \indentlabel{\@ifnotempty{#2}{
    \ignorespaces#2. #3}}
  \indentlabel{\@ifempty{#2}{\ignorespaces #3}{}}
    \vspace{1.5pt}}

\renewcommand{\tocsubsubsection}[3]{%
  \indentlabel{\@ifnotempty{#2}{
    \ignorespaces#2. #3}}
  \indentlabel{\@ifempty{#2}{\ignorespaces #3}{}}
    \vspace{1.5pt}}

\makeatother

\makeatletter
\def\@nomenstarted{0}
\newlength{\@nomenoldtabcolsep}

\newcommand{\nomenstart}
  {%
    \def\@nomenstarted{1}%
    \setlength{\@nomenoldtabcolsep}{\tabcolsep}%
    \setlength{\tabcolsep}{3.5pt}%
    \begin{longtable}{p{0.11\textwidth} p{0.86\textwidth}}
  }

\newcommand{\nomenitem}[2]{%
    \ifcase\@nomenstarted%
      \or 
      \or \\ 
    \fi%
    #1\,{\leavevmode\leaders\hbox{\,.}\hfill} & #2%
    \def\@nomenstarted{2}%
  }%
\newcommand{\nomenend}
  {\\%
      \end{longtable}%
      \setlength{\tabcolsep}{\@nomenoldtabcolsep}%
      \def\@nomenstarted{0}%
  }
\makeatother

\makeatletter
\newcommand{\vast}{\bBigg@{4}}
\newcommand{\Vast}{\bBigg@{5}}
\makeatother

\allowdisplaybreaks

\numberwithin{equation}{section}
\numberwithin{figure}{section}
\newtheorem{thm}{Theorem}[section]

\newtheorem{prop}[thm]{Proposition}
\newtheorem{lemma}[thm]{Lemma}

\newtheorem{prob}[thm]{Problem}
\newtheorem*{thm*}{Theorem}
\newtheorem*{prop*}{Proposition}
\newtheorem*{cor*}{Corollary}
\newtheorem*{conj*}{Conjecture}

\theoremstyle{definition}
\newtheorem{definition}[thm]{Definition}

\theoremstyle{remark}
\newtheorem{rmk}[thm]{Remark}

\newcommand{\mc}{\mathcal}
\newcommand{\cA}{\mc A}

\newcommand{\cC}{\mc C}

\newcommand{\cF}{\mc F}

\newcommand{\cH}{\mc H}

\newcommand{\cO}{\mc O}

\newcommand{\cX}{\mc X}

\newcommand{\ms}{\mathscr}

\newcommand{\sC}{\ms C}

\newcommand{\TT}{\mathbb{T}}

\newcommand{\C}{\mathbb{C}}
\newcommand{\N}{\mathbb{N}}
\newcommand{\R}{\mathbb{R}}
\newcommand{\Z}{\mathbb{Z}}

\newcommand{\Sph}{\mathbb{S}}

\newcommand{\fm}{\mathfrak{m}}

\newcommand{\ran}{\operatorname{ran}}

\renewcommand{\Re}{\operatorname{Re}}
\renewcommand{\Im}{\operatorname{Im}}

\newcommand{\mathspan}{\operatorname{span}}
\newcommand{\supp}{\operatorname{supp}}

\newcommand{\dS}{{\mathrm{dS}}}
\newcommand{\QNM}{{\mathrm{QNM}}}

\newcommand{\eps}{\epsilon}

\newcommand{\hra}{\hookrightarrow}
\newcommand{\la}{\langle}

\newcommand{\ol}{\overline}
\newcommand{\pa}{\partial}
\newcommand{\ra}{\rangle}

\newcommand{\weakto}{\rightharpoonup}

\newcommand{\wt}{\widetilde}

\newcommand{\pfstep}[1]{$\bullet$\ \textbf{#1}}

\newcommand{\dd}{{\mathrm{d}}}

\newcommand{\bop}{{\mathrm{b}}}

\newcommand{\cp}{{\mathrm{c}}}

\newcommand{\half}{{\tfrac{1}{2}}}

\newcommand{\CI}{\cC^\infty}

\newcommand{\CIc}{\cC^\infty_\cp}

\newcommand{\Hb}{H_{\bop}}

\newcommand{\bhm}{\fm}

\newcommand{\openbigpmatrix}[1]
  {%
    \def\@bigpmatrixsize{#1}%
    \addtolength{\arraycolsep}{-#1}%
    \begin{pmatrix}%
  }
\newcommand{\closebigpmatrix}
  {%
    \end{pmatrix}%
    \addtolength{\arraycolsep}{\@bigpmatrixsize}%
  }

\newlength{\enummargin}

\newcommand{\usref}[1]{{\upshape\ref{#1}}}

\DeclareGraphicsExtensions{.mps}

\makeatletter
\newcounter{@enumsave}

\newcommand{\ctrset}{\setcounter{enumi}{\the@enumsave}}
\makeatother

\makeatletter
\newcommand*{\fwbw}[1]{\expandafter\@fwbw\csname c@#1\endcsname}
\newcommand*{\@fwbw}[1]{\ifcase #1 \or {\rm fw}\or {\rm bw}\fi}
\AddEnumerateCounter{\fwbw}{\@fwbw}
\makeatother

\begin{document}

\title{Quasinormal modes of small Schwarzschild--de~Sitter black holes}

\begin{abstract}
  We study the behavior of the quasinormal modes (QNMs) of massless and massive linear waves on Schwarzschild--de~Sitter black holes as the black hole mass tends to $0$. Via uniform estimates for a degenerating family of ODEs, we show that in bounded subsets of the complex plane and for fixed angular momenta, the QNMs converge to those of the static model of de~Sitter space. Detailed numerics illustrate our results and suggest a number of open problems.
\end{abstract}

\date{\today}

\author{Peter Hintz}
\address{Department of Mathematics, Massachusetts Institute of Technology, Cambridge, Massachusetts 02139-4307, USA}
\email{phintz@mit.edu}

\author{YuQing Xie}
\address{Department of Physics, Massachusetts Institute of Technology, Cambridge, Massachusetts 02139-4307, USA}
\email{xyuqing@mit.edu}

\maketitle

\section{Introduction}
\label{SI}

We present a study of the quasinormal mode (QNM) spectrum of Schwarzschild--de~Sitter (SdS) black holes with very small mass $\bhm$, that is, the radius of the event horizon is very small compared to the radius of the cosmological horizon. The limit $\bhm\searrow 0$ is a singular one: while naively one obtains (the static model of) de~Sitter space as the limiting spacetime for $\bhm=0$, the topology changes in a discontinuous manner as the black hole disappears.

We supplement our results with numerics. Building on code developed by Jansen \cite{JansenMathematica}, we demonstrate how to accurately compute QNMs and mode solutions for black holes with extremely small mass using a carefully chosen discretization scheme. We also numerically compute the \emph{dual resonant states} (sometimes called \emph{co-resonant states}) which appear in formulas for computing numerical coefficients in quasinormal mode expansions, see~\S\ref{SsIW}; they are typically rather singular distributions. To our knowledge, this is the first numerical study of these in the context of black hole perturbation theory.

In order to explain our results in more detail, fix the cosmological constant $\Lambda>0$. In static coordinates, the metric on an SdS black hole with mass $0<\bhm<(9\Lambda)^{-1/2}$ is
\begin{equation}
\label{EqISdS}
  g_{\bhm,\Lambda} = -\mu_{\bhm,\Lambda}(r)\,\dd t^2 + \mu_{\bhm,\Lambda}(r)^{-1}\,\dd r^2 + r^2 g_{\Sph^2},\qquad
  \mu_{\bhm,\Lambda}(r) := 1-\frac{2\bhm}{r}-\frac{\Lambda r^2}{3}.
\end{equation}
Here $t\in\R$, while $r$ lies in the interval $(r_+(\bhm),r_\cC(\bhm))\subset(0,\infty)$ where $r_+(\bhm)<r_\cC(\bhm)$ are the positive roots of $\mu_{\bhm,\Lambda}$; lastly, $g_{\Sph^2}$ denotes the standard metric on the unit sphere $\Sph^2$. For small $\bhm>0$, one has $r_+(\bhm)/(2\bhm)\approx 1$ and $r_\cC(\bhm)\approx\sqrt{3/\Lambda}$ (see Lemma~\ref{LemmaCMRoots}). Let $\nu\in\C$. A \emph{quasinormal mode} or \emph{resonance} of the Klein--Gordon operator $\Box_{g_{\bhm,\Lambda}}-\nu$ is then a complex number $\omega\in\C$ for which there exists a nonzero \emph{mode solution} or \emph{resonant state} $u(r,\theta,\phi)$, i.e.\ a solution of the equation\footnote{We use the sign convention $\Box_g=|g|^{-1/2}\pa_i|g|^{1/2}g^{i j}\pa_j$.}
\[
  (\Box_{g_{\bhm,\Lambda}}-\nu)(e^{-i\omega t}u)=0
\]
with outgoing boundary conditions at $r=r_+(\bhm),r_\cC(\bhm)$. This means that near $r=r_\bullet(\bhm)$, $\bullet=+,\cC$, the function $u$ takes the form
\begin{equation}
\label{EqIOutgoing}
  u(r,\theta,\phi) = |r-r_\bullet(\bhm)|^{-i\omega/2\kappa_\bullet(\bhm)}u_\bullet(r,\theta,\phi),\quad
  \kappa_\bullet(\bhm):=\half|\mu_{\bhm,\Lambda}'(r_\bullet)|,
\end{equation}
with $u_\bullet$ smooth down to $r=r_\bullet$; here $\kappa_+(\bhm),\kappa_\cC(\bhm)$ are the surface gravities of the event and cosmological horizon, respectively.

It is well-known \cite{SaBarretoZworskiResonances,BonyHaefnerDecay,MelroseSaBarretoVasyResolvent,VasyMicroKerrdS} that the set
\[
  \QNM(\bhm,\Lambda,\nu)\subset\C
\]
of QNMs is discrete. We are interested in its behavior as $\bhm\searrow 0$ while $\Lambda,\nu$ are held fixed. Exploiting the spherical symmetry of $g_{\bhm,\Lambda}$, we shall study, for fixed $l\in\N_0$, the set
\[
  \QNM_l(\bhm,\Lambda,\nu) \subset \QNM(\bhm,\Lambda,\nu)
\]
of QNMs for which there exists a nonzero mode solution $u$ with angular momentum $l$, i.e.\ $u$ is of the separated form $u(r,\theta,\phi)=u_0(r)Y(\theta,\phi)$ where $Y$ is a degree $l$ spherical harmonic (that is, $\Delta_{\Sph^2}Y=-l(l+1)Y$).

Writing $B_0(R)=\{x\in\R^3\colon|x|<R\}$ and $r_\cC(0):=\sqrt{3/\Lambda}$, the \emph{static model of de~Sitter space} is the spacetime $\R_t\times B_0(r_\cC(0))$ equipped with the metric
\begin{equation}
\label{EqIdS}
  g_{\dS,\Lambda} := -\mu_{0,\Lambda}(r)\,\dd t^2 + \mu_{0,\Lambda}(r)^{-1}\,\dd r^2+r^2 g_{\Sph^2}
\end{equation}
in polar coordinates $x=r\omega$; here $\mu_{0,\Lambda}(r)=1-\frac{\Lambda r^2}{3}$. A complex number $\omega\in\C$ is a QNM of $\Box_{g_{\dS,\Lambda}}-\nu$ if there exists a mode solution $u(x)$ which is smooth in $|x|<r_\cC(0)$ and takes the form~\eqref{EqIOutgoing} near $r=r_\cC(0)$, with $\kappa_\cC(0)=\kappa_\dS=\sqrt{\Lambda/3}$. We compute the dS QNM spectrum in~\S\ref{SdS} using methods based on \cite{VasyWaveOndS} and \cite[Appendix~C]{HintzVasyKdSStability}; see also \cite{HintzXiedS}.

Our main result, illustrated in Figure~\ref{FigIQNM}, is the following.

\begin{thm}
\label{ThmIMain}
  Fix $\Lambda>0$, $\nu\in\C$, and $l\in\N_0$, and set $\lambda_\pm:=\tfrac32\pm\sqrt{\tfrac94-\nu}$, $\kappa_\dS:=\sqrt{\Lambda/3}$.
  \begin{enumerate}
  \item\label{ItIMainQNM}{\rm (Convergence of QNMs).} Then the set $\QNM_l(\bhm,\Lambda,\nu)$ converges locally uniformly to the set
  \[
    \QNM_{\dS,l}(\Lambda,\nu) := \bigcup_\pm\,\bigl\{ -i(\lambda_\pm+l+2 j)\kappa_\dS \colon j\in \N_0 \bigr\}.
  \]
  This is the set of QNMs, with mode solutions of angular momentum $l$, for the operator $\Box_{g_{\dS,\Lambda}}-\nu$.
  \item\label{ItIMainMS}{\rm (Convergence of mode solutions.)} Let $\omega_0\in\QNM_{\dS,l}(\Lambda,\nu)$ and fix a degree $l$ spherical harmonic $Y(\theta,\phi)$. Let $\omega_\bhm\in\QNM_l(\bhm,\Lambda,\nu)$ with $\lim_{\bhm\searrow 0}\omega_\bhm=\omega_0$. Then the space of mode solutions of the form $e^{-i\omega t}Y(\theta,\phi)u_0(r)$ on de~Sitter space is 1-dimensional, and the same holds true for the space of outgoing mode solutions $e^{-i\omega_\bhm t}Y(\theta,\phi)u_\bhm(r)$ of $\Box_{g_{\bhm,\Lambda}}-\nu$ for all sufficiently small $\bhm>0$. After rescaling the $u_\bhm$ by suitable constants, we then have convergence $u_\bhm(r)\to u_0(r)$ locally uniformly, together with all derivatives, for $r\in(0,\sqrt{3/\Lambda})$.
  \end{enumerate}
\end{thm}

For massless scalar waves ($\nu=0$), a rescaling of this result provides control of QNMs when the black hole mass is fixed but the cosmological constant tends to $0$; see Remark~\ref{RmkINSmallLambda}.

As we recall in~\S\ref{SsIW}, QNMs and the corresponding mode solutions control the late-time asymptotics of linear waves; in this context, part~\eqref{ItIMainQNM} confirms the numerical results obtained by Brady--Chambers--Laarakeers--Poisson \cite[\S{III.B}]{BradyChambersLaarakeersPoissonSdSFalloff}. Thus, the significance of Theorem~\ref{ThmIMain} is that it demonstrates (for bounded frequencies, though see Remark~\ref{RmkIStronger}) that the late-time behavior is dominated by the large scale, de~Sitter, structure of the spacetime, and the exponential decay rate is essentially given by the surface gravity of the de~Sitter (cosmological) horizon. The presence of the black hole is barely visible at this level; according to our numerics, $\omega_\bhm/\sqrt\Lambda$ differs from the dS value $\omega_0/\sqrt\Lambda$ only by terms of size $\bhm^2\Lambda$ when $\nu=0$, and by terms of size $\bhm\sqrt\Lambda$ for generic $\nu\in\C$, see \S\ref{SsNR}. This provides corrections to the results of \cite{BradyChambersKrivanLagunaCosmConst} obtained by direct integration of the wave equation. The QNMs associated with the black hole and its \emph{ringdown} are not visible; see also the discussion around~\eqref{EqIPhotonSphere} below. For a detailed review of the role of QNMs in gravitational wave physics, we refer the reader to \cite{BertiCardosoStarinetsQNM}, see also the more recent \cite{IsiGieslerFarrScheelTeukolskyNoHair}. The QNM spectrum of small black holes in \emph{anti}-de~Sitter spacetimes was studied by Horowitz--Hubeny \cite{HorowitzHubenyAdSQNM} and Konoplya \cite{KonoplyaSmallSAdS}, motivated by the AdS-CFT correspondence. A detailed numerical analysis of the QNM spectrum of Reissner--Nordstr\"om--de~Sitter black holes (the charged generalization of SdS black holes) was performed in \cite{CardosoCostaDestounisHintzJansenSCC}.

\begin{figure}[!ht]
\centering
\includegraphics[width=\textwidth]{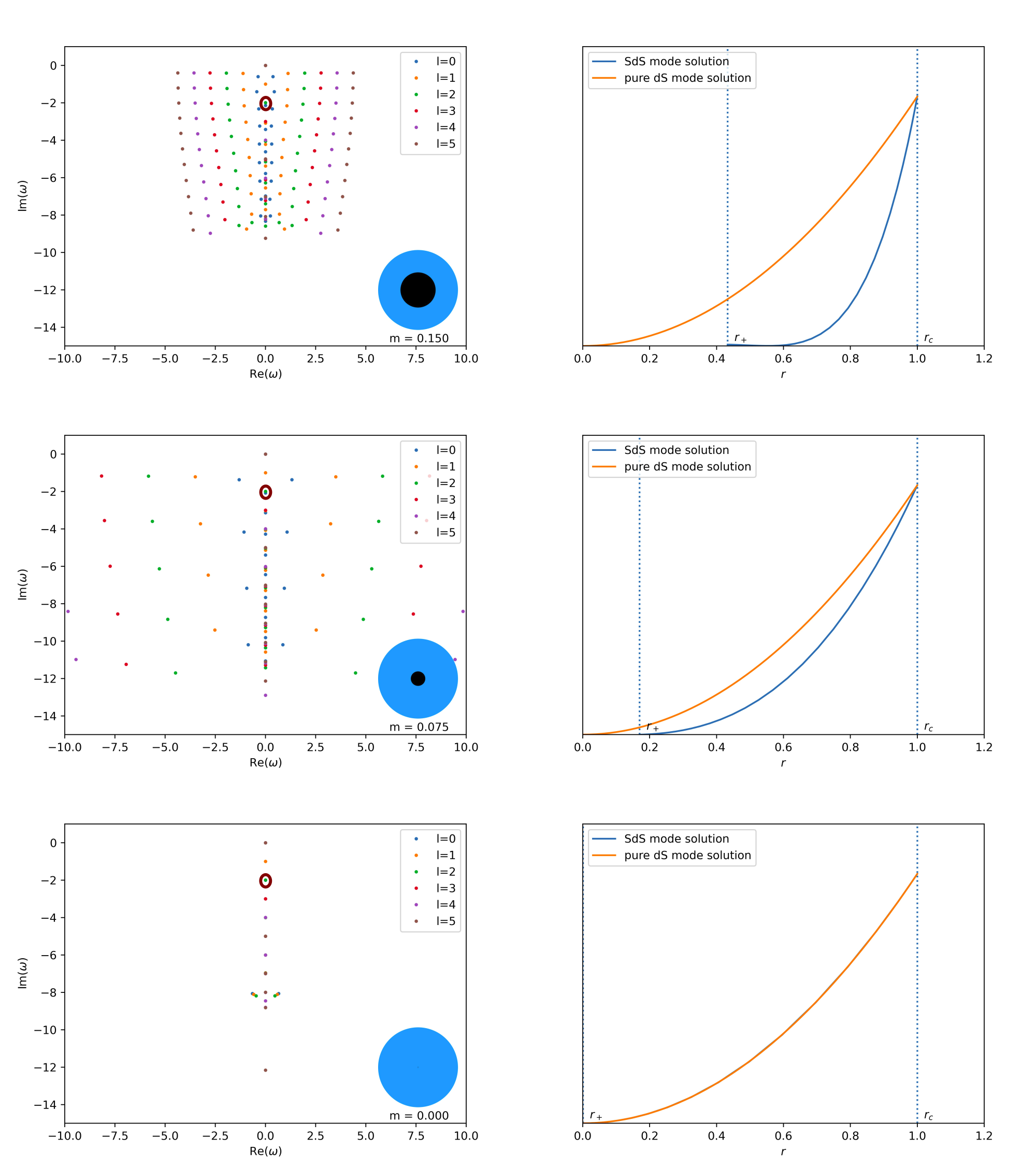}
\caption{\textit{Left:} QNMs of a SdS black hole for $\Lambda=3$, $\nu=0$, and $\bhm=0.150,0.075,0.000$ (from top to bottom). \textit{Right:} mode solution with angular momentum $l=2$ corresponding to the highlighted QNM on the left. The domain is rescaled by $r_\cC(\bhm)$ so that the cosmological horizon always lies at $r_c=1$.}
\label{FigIQNM}
\end{figure}

\begin{rmk}[Stronger statements]
\label{RmkIStronger}
  The simple setting of Theorem~\ref{ThmIMain} allows for a proof by means of elementary ODE estimates, which we will however phrase in a readily generalizable manner. In particular, the restriction to fixed angular momenta can easily be removed by using the red-shift estimates from \cite{WarnickQNMs} in the SdS and \cite{VasyMicroKerrdS} in the Kerr--de~Sitter (KdS) setting. More strongly still, we expect that the QNMs in fact converge in any half space $\Im\omega>-C$; this however requires significantly more work, and will be discussed elsewhere.
\end{rmk}

The QNM spectrum of SdS black holes with fixed mass $\bhm\in(0,(9\Lambda)^{-1/2})$ has been studied in detail in the high energy regime, i.e.\ when $|\Re\omega|\gg 1$, where geometric optics approximations and semiclassical methods near the photon sphere apply. QNMs in any strip $\Im\omega\geq -C$ and for large angular momenta $l$ are well approximated by
\begin{equation}
\label{EqIPhotonSphere}
  \Bigl(\pm l\pm\frac12-\frac{i}{2}\Bigl(n+\frac12\Bigr)\Bigr)\frac{(1-9\Lambda\bhm^2)^{1/2}}{3\bhm\sqrt{3}},\qquad n=0,1,2,\ldots,
\end{equation}
as proved by S\'a Barreto--Zworski \cite{SaBarretoZworskiResonances}; see also \cite{ZhidenkoQNMofSdS}. Dyatlov \cite{DyatlovAsymptoticDistribution} gave a full asymptotic expansion of high energy QNMs for KdS black holes, see also \cite{NovaesMarinhoLencsesKdSQNM}. The Schwarzschild case was previously studied by Bachelot and Motet-Bachelot \cite{BachelotSchwarzschildScattering,BachelotMotetBachelotSchwarzschild}. As $\bhm\searrow 0$, the QNMs approximated by~\eqref{EqIPhotonSphere} scale like $\bhm^{-1}$, explaining their absence in the statement of Theorem~\ref{ThmIMain}.

We proceed to explain Theorem~\ref{ThmIMain} and this scaling by describing the two different limits of $g_{\bhm,\Lambda}$ as $\bhm\searrow 0$. For the first limit, fix $r_0>0$; then for $r>r_0$, the coefficients of $g_{\bhm,\Lambda}$ converge to those of $g_{\dS,\Lambda}$ as $\bhm\searrow 0$. On the other hand, for any small $\bhm>0$, the metric $g_{\bhm,\Lambda}$ differs from $g_{\dS,\Lambda}$ by terms of unit size (and so do e.g.\ curvature scalars) close to the black hole. In terms of the rescaled time and radius functions
\begin{equation}
\label{EqIResc}
  \hat t := \frac{t}{\bhm},\quad
  \hat r := \frac{r}{\bhm};
\end{equation}
we have
\[
  \bhm^{-2}g_{\bhm,\Lambda} = -\hat\mu_{\bhm,\Lambda}(\hat r)\,\dd\hat t^2 + \hat\mu_{\bhm,\Lambda}(\hat r)^{-1}\,\dd\hat r + \hat r^2 g_{\Sph^2},\qquad
  \hat\mu_{\bhm,\Lambda}(\hat r) := 1-\frac{2}{\hat r}-\bhm^2\frac{\Lambda\hat r^2}{3},
\]
which for bounded $\hat r$ converges to the metric
\[
  \hat g_1 := -\Bigl(1-\frac{2}{\hat r}\Bigr)\,\dd\hat t^2 + \Bigl(1-\frac{2}{\hat r}\Bigr)^{-1}\,\dd\hat r^2 + \hat r^2 g_{\Sph^2}
\]
of a Schwarzschild black hole with mass $1$.

Corresponding to these two regimes then, the analysis of the spectral family
\begin{equation}
\label{EqISdSSpec}
  P_{\bhm,\Lambda,\nu}(\omega) := e^{i\omega t}(\Box_{g_{\bhm,\Lambda}}-\nu)e^{-i\omega t}
\end{equation}
of the Klein--Gordon operator (i.e.\ formally replacing $\pa_t$ by $-i\omega$ in the expression for $\Box_{g_{\bhm,\Lambda}}$) for fixed $\omega$ and small $\bhm$ requires the analysis of two limiting models. The first model is the spectral family of $\Box_{g_{\dS,\Lambda}}-\nu$ at frequency $\omega$. The second model arises by considering
\[
  \bhm^2 e^{i\omega t}(\Box_{g_{\bhm,\Lambda}}-\nu)e^{-i\omega t}\approx e^{i\bhm\omega\hat t}(\Box_{\hat g_1}-\bhm^2\nu)e^{-i\bhm\omega\hat t}
\]
and formally taking $\bhm=0$ on the right hand side; this gives the restriction of the \emph{massless} wave operator $\Box_{\hat g_1}$ to \emph{zero energy} (i.e.\ $\hat t$-independent) functions. Note that a QNM $\hat\omega$ of the mass $1$ Schwarzschild metric should roughly correspond to a QNM
\begin{equation}
\label{EqIOmegaResc}
  \omega=\bhm^{-1}\hat\omega
\end{equation}
of the mass $\bhm$ SdS metric, in alignment with~\eqref{EqIPhotonSphere}.

The proof of Theorem~\ref{ThmIMain} combines a priori estimates for the ordinary differential operators arising from the spherical harmonic decomposition of $e^{i\omega_0 t}(\Box_{g_{\dS,\Lambda}}-\nu)e^{-i\omega_0 t}$ and of the zero energy operator of $\Box_{\hat g_1}$. If \emph{both} operators are invertible on suitable outgoing function spaces, then so is the SdS operator~\eqref{EqISdSSpec} for $\omega$ near $\omega_0$ and $\bhm>0$ sufficiently small. See Figure~\ref{FigIBlowup} for an illustration. Note that the de~Sitter model is still singular in that $r=0$ is removed, which on an analytic level amounts to working with function spaces which encode suitable boundary conditions at $r=0$ (roughly, restricting the growth as $r\to 0$); we show that for the correct choice of boundary condition, the resulting set of QNMs (and mode solutions) agrees with the usual one as defined  after~\eqref{EqIdS}. We furthermore remark that the outgoing condition at the cosmological horizon of SdS limits to that on dS, whereas the outgoing condition at the event horizon survives in the Schwarzschild model problem.

The invertibility of the zero energy operator on Schwarzschild is proved in~\S\ref{SsCPf}, and that of the de~Sitter model holds if and only if $\omega_0$ is not a dS QNM. If however $\omega_0$ \emph{is} a dS QNM, then an argument based on Rouch\'e's theorem proves the existence of a QNM $\omega_\bhm$ of SdS with mass $\bhm$ with $\omega_\bhm\to\omega_0$ as $\bhm\searrow 0$.

\begin{figure}[!ht]
\centering
\includegraphics{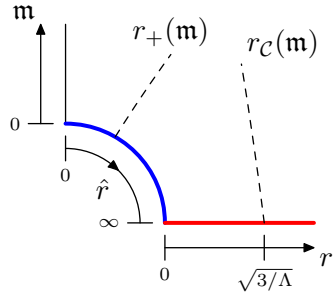}
\caption{Illustration of the way in which the two limiting regimes---$\bhm\searrow 0$ in $r\gtrsim 1$ (red) and $\hat r=\frac{r}{\bhm}\sim 1$ (blue)---fit together, and the location of the event horizon and the cosmological horizon (dashed).}
\label{FigIBlowup}
\end{figure}

\begin{rmk}[Modified time function and the outgoing condition]
\label{RmkItstar}
  In our proofs, we will work with a time function $t_*=t-F(r)$ for a suitably chosen function $F$ (with logarithmic singularities at $r_+(\bhm),r_\cC(\bhm)$) so that the level sets of $t_*$ are smooth and transversal to the future event and cosmological horizon; see Figure~\ref{FigIPenrose} and~\S\ref{SsCM}. QNMs are then characterized as those $\omega\in\C$ so that $(\Box_{g_{\bhm,\Lambda}}-\nu)(e^{-i\omega t_*}u)=0$ for some $u\neq 0$ which is smooth on an enlarged interval $[\half r_+(\bhm),2 r_\cC(\bhm)]$, i.e.\ $u$ is smooth across both horizons. We shall prove smooth convergence of mode solutions on such an enlarged interval, thus strengthening part~\eqref{ItIMainMS} of Theorem~\ref{ThmIMain}.
\end{rmk}

\begin{figure}[!ht]
\centering
\includegraphics{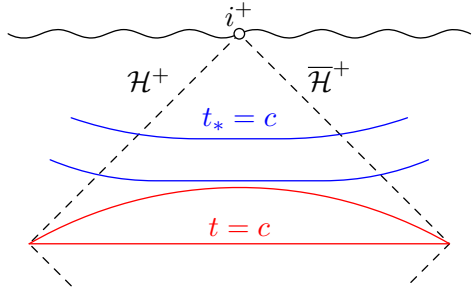}
\caption{Penrose diagram of a SdS black hole spacetime, including the future event horizon $\cH^+=\{r=r_+(\bhm)\}$, the future cosmological horizon $\ol\cH^+=\{r=r_\cC(\bhm)\}$, future timelike infinity $i^+$, and the future conformal boundary and black hole singularity (squiggly). Also shown are two level sets each of the static time coordinate $t$ (red) and the modified time coordinate $t_*$ (blue).}
\label{FigIPenrose}
\end{figure}

\subsection{Klein--Gordon equations and dual resonant states}
\label{SsIW}

Knowledge of the QNM spectrum in a half space $\Im\omega>-C$ is of critical importance for understanding the long-time dynamics of solutions of the Klein--Gordon equation. Denoting QNMs and mode solutions by $\omega_j,u_j$, $j=0,1,2,\ldots$, and ignoring the possibility of multiplicity for simplicity, the solution of the initial value problem
\begin{equation}
\label{EqIVP}
  (\Box_{g_{\bhm,\Lambda}}-\nu)u(t_*,r,\theta,\phi) = 0,\qquad(u,\pa_{t_*} u)|_{t_*=0}=(u^0,u^1),
\end{equation}
with initial data $u^0,u^1$ smooth across the horizons, can then be expanded as
\begin{equation}
\label{EqIExp}
  u(t_*,r,\theta,\phi) = \sum_{\Im\omega_j>-C} e^{-i\omega_j t_*} c_j u_j(r,\theta,\phi) + \cO(e^{-(C-\eps)t_*}),\qquad t_*\to\infty,
\end{equation}
for any $C\in\R$ and $\eps>0$ \cite{BonyHaefnerDecay,DyatlovQNMExtended}; the coefficients $c_j$ are complex numbers which are computable from the initial data $(u^0,u^1)$, as discussed further below. The relationship between exponential decay in the exterior, specifically the quantity $\min(-\Im\omega_j)$, and the blue-shift effect in the interior of charged or rotating black holes\footnote{This excludes SdS black holes; but even in the realm of spherical symmetry, our results admit straightforward extensions to the case of charged Reissner--Nordstr\"om--de~Sitter black holes.} has an immediate bearing on Penrose's Strong Cosmic Censorship conjecture \cite{PenroseSCC} in the context of cosmological black hole spacetimes, as discussed in detail in \cite{HintzVasyCauchyHorizon,CardosoCostaDestounisHintzJansenSCC,DiasEperonReallSantosSCC} following the earlier works \cite{MellorMossStability,BradyMossMyersCosmicCensorship,BradyChambersLaarakeersPoissonSdSFalloff}. Adopting the terminology of \cite{CardosoCostaDestounisHintzJansenSCC}, the modes appearing in Theorem~\ref{ThmIMain} are the \emph{de~Sitter modes}, whereas the \emph{photon sphere modes} approximated by~\eqref{EqIPhotonSphere} leave every compact subset of $\C$ as $\bhm\searrow 0$. We also note that control of $\min(-\Im\omega_j)$ and of the resonant states $u_j$ is important in nonlinear applications, see in particular \cite{HintzVasyKdSStability}.

The coefficient $c_j$ in the expansion~\eqref{EqIExp} of the solution $u$ of a forcing problem $(\Box_{g_{\bhm,\Lambda}}-\nu)u=f$, $(u,\pa_t u)|_{t_*=0}=(0,0)$, and with $f$ of compact support in $t_*$, is given explicitly by
\begin{equation}
\label{EqICoeffs}
  c_j = \frac{i\int_\R \la f(t_*,-), e^{-i\ol{\omega_j}t_*}u_j^*\ra\,\dd t_*}{\la\pa_\omega P_{\bhm,\Lambda,\nu}(\omega_j)u_j,u_j^*\ra},
\end{equation}
where $\la-,-\ra$ denotes the $L^2$ inner product on $t_*=const.$ with integration measure $r^2\sin\theta\,\dd r\,\dd\phi\,\dd\theta$, and $P_{\bhm,\Lambda,\nu}(\omega)$ is defined as in~\eqref{EqISdSSpec} but with $t_*$ in place of $t$. (For the initial value problem~\eqref{EqIVP}, one uses $f:=[\Box,H(t_*)](u_0+t_* u_1)$ in~\eqref{EqICoeffs}.) The crucial ingredient here is the \emph{dual resonant state} $u_j^*=u_j^*(r,\theta,\phi)$: this is a distributional solution of the adjoint equation $P_{\bhm,\Lambda,\nu}(\omega_j)^*u_j^*=0$, or equivalently
\[
  (\Box_{g_{\bhm,\Lambda}}-\bar\nu)(e^{-i\ol{\omega_j}t_*}u_j^*)=0.
\]
See \cite{HintzXiedS} for a qualitative discussion of dual resonant states. The distribution $u_j^*$ vanishes for $r<r_+(\bhm)$ and $r>r_\cC(\bhm)$ and is singular at $r=r_+(\bhm)$ and $r=r_\cC(\bhm)$, see e.g.\ \cite[\S4.2]{HintzVasyKdsFormResonances} and the proof of Proposition~\ref{PropdSFred} below. In~\S\ref{SsNR}, we compute the dual resonant states $u_j^*$ numerically and investigate their convergence to the dual resonant states on de~Sitter space as $\bhm\searrow 0$.

\subsection{Normalization}
\label{SsIN}

Let $s>0$, then the pullback of $r\mapsto\mu_{\bhm,\Lambda}(r)$ under the dilation map $M_s\colon(t,r)\mapsto(s t,s r)$ is $r\mapsto\mu_{\bhm/s,\Lambda s^2}(r)$, thus $M_s^*g_{\bhm,\Lambda}=s^2 g_{\bhm/s,\Lambda s^2}$ and correspondingly
\[
  (\Box_{g_{\bhm,\Lambda}}-\nu)(e^{-i\omega t}a(r,\theta,\phi))=0
  \ \Longleftrightarrow\ 
  (\Box_{g_{\bhm/s,\Lambda s^2}}-s^2\nu)(e^{-i(s\omega)t}a(s r,\theta,\phi))=0.
\]
Therefore, $\omega$ is a QNM for the Klein--Gordon equation with parameters $(\bhm,\Lambda,\nu)$ if and only if $s\omega$ is a QNM for the parameters $(\bhm/s,\Lambda s^2,s^2\nu)$. For $s=\sqrt{3/\Lambda}$, this gives
\begin{equation}
\label{EqIN}
  \QNM(\bhm,\Lambda,\nu)=\sqrt{\frac{\Lambda}{3}} \QNM\Bigl(\bhm\sqrt{\frac{\Lambda}{3}},3,\frac{3\nu}{\Lambda}\Bigr).
\end{equation}
We may thus fix the value
\[
  \Lambda=3
\]
throughout the remainder of the paper, and simply write
\[
  g_\bhm := g_{\bhm,3},\qquad
  g_\dS := g_{\dS,3},\qquad
  \mu_\bhm := \mu_{\bhm,3}=1-\frac{2\bhm}{r}-r^2.
\]
Thus, $\mu_0(r)=1-r^2$ and $\kappa_\cC(0)=\kappa_\dS=1$.

\begin{rmk}[Small $\Lambda$ for massless waves]
\label{RmkINSmallLambda}
  When $\nu=0$, the relationship~\eqref{EqIN} becomes $\QNM(\bhm,\Lambda,0)=\sqrt{\frac{\Lambda}{3}}\QNM(\bhm\sqrt{\frac{\Lambda}{3}},3,0)$. Thus, if we fix $\bhm>0$ but let $\Lambda\searrow 0$, then Theorem~\ref{ThmIMain} gives control of SdS QNMs in balls of radius $C\sqrt{\Lambda}$ for any fixed $C>0$. (This is the physically more relevant setting, as according to the currently favored $\Lambda$CDM model of cosmology, the cosmological constant is indeed very small and positive.)
\end{rmk}

\subsection{Plan of the paper}

In~\S\ref{SdS}, we compute the QNMs of the static model of de~Sitter space, giving details for the arguments sketched in \cite{HintzXiedS}. The mathematical heart of the paper is~\S\ref{SC}, where we prove Theorem~\ref{ThmIMain}. Numerics are presented in~\S\ref{SN}, with a number of open problems stated in~\S\ref{SsNR}.

\subsection*{Acknowledgments}

PH gratefully acknowledges support from the NSF under Grant No.\ DMS-1955614 and from a Sloan Research Fellowship. Part of this research was conducted during the time PH served as a Clay Research Fellow. YX would like to thank the MIT UROP office for this opportunity and gratefully acknowledges support from the Reed Fund, Baruch Fund, and Anonymous Fund.

\section{QNMs of de~Sitter space}
\label{SdS}

For $\Lambda=3$, the de~Sitter metric given by~\eqref{EqIdS} can also be expressed as
\begin{subequations}
\begin{align}
\label{EqIdSCosm}
  g_\dS &= -\dd t_*^2 + e^{-2 t_*}\dd X^2 \\
\label{EqIdSConf}
    &= \frac{-\dd\tau^2+\dd X^2}{\tau^2} \\
\label{EqIdSStaticExt}
    &= -(1-r^2)\dd t_*^2 - 2 r\,\dd r\,\dd t_* + \dd r^2+r^2 g_{\Sph^2};
\end{align}
\end{subequations}
the coordinates are related to each other via
\begin{equation}
\label{EqdSCoords}
  \tau=e^{-t_*},\quad
  x:=\frac{X}{\tau}=r\omega,\ \omega\in\Sph^2,\quad
  t=t_*-\half\log(1-r^2).
\end{equation}
We present some details here, in $3+1$ dimensions, of the arguments sketched in \cite{HintzXiedS} regarding the long-time behavior of solutions of the Klein--Gordon equation on de~Sitter space.

\begin{prop}[QNMs of de~Sitter space]
\label{PropdS}
  Let $\nu\in\C$ and $\lambda_\pm(\nu)=\frac{3}{2}\pm\sqrt{\frac{9}{4}-\nu}$, and recall the normalization $\Lambda=3$. Then the set of QNMs of $\Box_{g_\dS}-\nu$ is equal to
  \[
    \QNM_\dS(\Lambda,\nu) = \{ -i\lambda_\pm(\nu) - i j\colon j\in\N_0 \}.
  \]
  The set $\QNM_{\dS,l}(\Lambda,\nu)$ of QNMs permitting mode solutions with angular momentum $l$ is the subset where $j$ is restricted to $l,l+2,l+4,l+6,\ldots$. For $\omega\in\QNM_{\dS,l}(\Lambda,\nu)$, the space of angular momentum $l$ mode solutions has dimension $2 l+1$.
\end{prop}

As argued in \cite{HintzXiedS} (and as exploited in similar contexts already in \cite[\S4.1]{HintzVasyKdsFormResonances} and \cite[Appendix~C]{HintzVasyKdSStability}), this is a consequence of:

\begin{prop}[Asymptotic behavior of waves]
\label{PropdSGbl}
  Let $L>0$ and define the metric $g=\frac{-\dd\tau^2+\dd X^2}{\tau^2}$, where $\dd X^2$ is the standard metric on $\TT_L^3:=\R^3/(L\Z)^3$. Define $\lambda_\pm(\nu)$ as in Proposition~\usref{PropdS}. Let $u^0,u^1\in\CI(\TT_L^3)$, and let $u(\tau,X)$ denote the solution of
  \[
    (\Box_g-\nu)u = 0,\quad
    u(1,X)=u^0(X),\ -\tau\pa_\tau u(1,X)=u^1(X).
  \]
  Let $d=\lambda_+(\nu)-\lambda_-(\nu)$. Then there exist $u_\pm\in\CI([0,1]\times\TT^3_L)$ which are even in $\tau$ (i.e.\ $\pa_\tau^j u_\pm(0,X)\equiv 0$ for all odd $j\in\N$) so that
  \begin{equation}
  \label{EqdSGblAsy}
    u(\tau,X)
      =\begin{cases} 
         \tau^{\lambda_-(\nu)}u_-(\tau,X) + \tau^{\lambda_+(\nu)}u_+(\tau,X), & d\notin 2\N_0, \\
         \tau^{\lambda_-(\nu)}u_-(\tau,X) + \tau^{\lambda_+(\nu)}(\log\tau)u_+(\tau,X), & d\in 2\N_0.
       \end{cases}
  \end{equation}
  The map $(u^0,u^1)\mapsto(u_-,u_+)|_{\tau=0}$ is a linear isomorphism of $\CI(\TT^3_L)\oplus\CI(\TT^3_L)$. Moreover, the full Taylor series of $u_\pm$ at $\tau=0$ are uniquely determined by $u_\pm|_{\tau=0}$.
\end{prop}
\begin{proof}
  We begin by proving an estimate for the energy
  \[
    E[u](\tau):=\frac12\int_{\TT^3_L} \bigl(|\tau\pa_\tau u(\tau,X)|^2+|\tau\pa_X u(\tau,X)|^2+|u(\tau,X)|^2\bigr)\,\dd X
  \]
  as $\tau\searrow 0$, where $u$ solves the Klein--Gordon equation with forcing $f$,
  \begin{equation}
  \label{EqdSGbl}
    (\Box_g-\nu)u = \bigl(-(\tau\pa_\tau)^2 + 3\tau\pa_\tau + \tau^2\Delta_X - \nu\bigr)u = f.
  \end{equation}
  Direct differentiation of $E(\tau)$ and integration by parts give
  \begin{align*}
    -\tau\pa_\tau E[u] &= \Re\int_{\TT^3_L} \bigl(-(\tau\pa_\tau)^2 u\ol{\tau\pa_\tau u} - \tau\pa_X u\ol{\tau\pa_X\tau\pa_\tau u} - |\tau\pa_X u|^2 - u\ol{\tau\pa_\tau u}\bigr)\,\dd X \\
      &= \Re \int_{\TT^3_L} \bigl(-3|\tau\pa_\tau u|^2 - |\tau\pa_X u|^2 + (\nu-1)u\ol{\tau\pa_\tau u}\bigr)+f\ol{\tau\pa_\tau u}\,\dd X \\
      &\leq C E[u] + \half\|f(\tau,\cdot)\|_{L^2(\TT^3_L)}^2,
  \end{align*}
  where $C=|\nu-1|+1$. Therefore,
  \[
    \tau^C E[u](\tau)\leq E[u](1)+\int_\tau^1 \half\|\tilde\tau^C f(\tilde\tau,\cdot)\|_{L^2(\TT^3_L)}^2\,\frac{\dd\tilde\tau}{\tilde\tau}.
  \]
  For $f=0$ this gives $E[u](\tau)\leq\tau^{-C}E[u](1)$. Moreover, if $u$ solves~\eqref{EqdSGbl} with $f=0$, then so does $\Delta_X^N u$ for all $N\in\N_0$. We can furthermore estimate $\tau\pa_\tau u$ as the solution of $(\Box_g-\nu)(\tau\pa_\tau u)=-2\tau^2\Delta_X u$, giving $E[\tau\pa_\tau\Delta_X^N u]\lesssim\tau^{-C}$; higher derivatives in $\tau\pa_\tau$ are estimated similarly. Thus, by Sobolev embedding on $\TT^3_L$, we have
  \[
    u \in \cA^{-C}([0,1]_\tau;\CI(\TT^3_L)),
  \]
  by which we mean that for all $N,k\in\N_0$, one has $\|(\tau\pa_\tau)^k u\|_{\cC^N(\TT^3_L)}\leq C_{N k}\tau^{-C}$ for $\tau\in(0,1]$.

  The rest of the proof closely follows \cite[\S4]{VasyWaveOndS}. Rewrite equation~\eqref{EqdSGbl} with $f=0$ as
  \begin{equation}
  \label{EqdSGblCons}
    \bigl((\tau\pa_\tau)^2-3\tau\pa_\tau-\nu\bigr)u = \tau^2\Delta_X u.
  \end{equation}
  This is a regular-singular ODE in $\tau$, with indicial roots $\lambda_\pm(\nu)$ and right hand side in $\cA^{-C+2}([0,1];\CI(\TT_L^3))$. If $-C+2<\Re\lambda_\pm(\nu)$, we conclude that $u\in\cA^{-C+2}$, thus improving the decay as $\tau\to 0$ by a power of $\tau^2$; otherwise, if, say, $\Re\lambda_-(\nu)\in(-C,-C+2)$ but $\Re\lambda_+(\nu)\notin(-C,-C+2)$, one gets
  \begin{equation}
  \label{EqdSGblFirst}
    u=\tau^{\lambda_-(\nu)}u_-^0+\tilde u^0,\qquad
    u_-^0\in\CI(\TT^3_L),\ \tilde u^0\in\cA^{-C+2}([0,1];\CI(\TT^3_L)).
  \end{equation}
  For simplicity, we shall assume for the rest of the proof that $\lambda_+(\nu)-\lambda_-(\nu)\notin 2\N_0$, and leave the simple modifications for the general case to the reader. 

  At this point, it is convenient to pause and consider the \emph{construction} of a solution with prescribed leading order asymptotics. For definiteness, let us consider asymptotic data $u_-(0,X)=u_-^0(X)\in\CI(\TT^3_L)$ and $u_+(0,X)\equiv 0$, and show how to construct $u_-(\tau,X)\in\CI([0,1]\times\TT^3_L)$ so that $u(\tau,X)=\tau^{\lambda_-(\nu)}u_-(\tau,X)$ solves $(\Box_g-\nu)u=0$. This relies on the rewriting~\eqref{EqdSGblCons}: having constructed $u_-^{2 k}$ for $k=0,\ldots,k_0-1$, we define $u_-^{2 k}\in\CI(\TT^3_L)$ by
  \begin{equation}
  \label{EqdSGblConsIt}
    (\lambda^2-3\lambda-\nu)|_{\lambda=\lambda_-(\nu)+2 k}\cdot u_-^{2 k} = \Delta_X u_-^{2 k-2}.
  \end{equation}
  Choose $\hat u_-\in\CI([0,1]\times\TT^3_L)$ so that $\pa_\tau^{2 k}\hat u_-(0)/(2 k)!=u_-^{2 k}$ for all $k=0,1,2,\ldots$, then $\tau^{\lambda_-(\nu)}\hat u_-$ is an approximate solution of the Klein--Gordon equation in the sense that $f:=(\Box_g-\nu)(\tau^{\lambda_-(\nu)}\hat u_-)\in\cA^N([0,1];\CI)$ for all $N$. By solving $(\Box_g-\nu)u_{\tau_0}=f$ with initial data $(u_{\tau_0},-\tau\pa_\tau u_{\tau_0})|_{\tau=\tau_0}=(0,0)$ on $[\tau_0,1]\times\TT^3_L$ and letting $\tau_0\to 0$, one can show (using energy estimates similarly to above) that there exists $\tilde u\in\bigcap_N\cA^N([0,1];\CI)$ with $(\Box_g-\nu)\tilde u=f$. Hence, $u:=\tau^{\lambda_-(\nu)}u_-$ with $u_-=\hat u_--\tau^{-\lambda_-(\nu)}\tilde u\in\CI([0,1]\times\TT^3_L)$ solves $(\Box_g-\nu)u=0$, as desired. Note that the Taylor series of $u_-$ at $\tau=0$ is uniquely determined by $u_-^0$.

  Returning to the proof of~\eqref{EqdSGblAsy}, once we have obtained the first leading order term $\tau^{\lambda_-(\nu)}u_-^0$ as in~\eqref{EqdSGblFirst}, we denote by $\tau^{\lambda_-(\nu)}u_-$ (with $u_-\in\CI([0,1]\times\TT^3_L)$) the solution of the Klein--Gordon equation with asymptotic data $(u_-^0,0)$; then $\tilde u:=u-\sum_-\tau^{\lambda_-(\nu)}u_-$ is a solution with trivial $\tau^{\lambda_-(\nu)}$ leading order term, so $\tilde u\in\cA^{-C+2}$. One can then proceed with the iterative argument based on~\eqref{EqdSGblCons} until one encounters the indicial root $\lambda_+(\nu)$. One then subtracts a solution with the corresponding leading order term from $\tilde u$ to obtain a solution $\tilde u'\in\cA^{C'}$ with $C'>\Re\lambda_\pm(\nu)$. But then one can conclude $\tilde u'\in\cA^N$ for $N=C',C'+2,C'+4,\ldots$, thus $\tilde u'$ vanishes to infinite order at $\tau=0$, and then an energy estimate shows that $\tilde u'$ must, in fact, vanish (see also \cite[Lemma~1]{ZworskiRevisitVasy}). This completes the proof.
\end{proof}

\begin{proof}[Proof of Proposition~\usref{PropdS}]
  The starting point is the fact that the set of resonances for the Klein--Gordon equation on dS in a half space $\Im\omega>-C$ is finite for all $C\in\R$ since dS is nontrapping, see \cite[\S4]{VasyMicroKerrdS} or \cite[\S2.1]{HintzVasySemilinear}. Moreover, any QNM $\omega_j$ (and its corresponding mode solution $u_j$) contributes nontrivially in the expansion~\eqref{EqIExp}, i.e.\ with $c_j\neq 0$, for suitable initial data $(u^0,u^1)$ (for instance, for $(u^0,u^1)=(u_j,-i\omega_j u_j)$). Thus, we shall determine the asymptotic behavior of $u$ as $t_*\to\infty$ by means of Proposition~\ref{PropdSGbl} via a change of coordinates, and read off the expansion~\eqref{EqIExp}, hence the QNMs and resonant states.

  The neighborhood of (part of) the static patch of de~Sitter space of interest to us is $D=[0,\infty)_{t_*}\times\{|x|<2\}$, with initial data for waves posed at $t_*=0$. Translated to $(\tau,X)=(e^{-t_*},x e^{-t_*})\in(0,1]\times\R^3$ (see~\eqref{EqdSCoords}), this domain is contained in $D':=\{\tau\leq 1,\,|X|<2\}$. By finite speed of propagation, solutions of the wave equation on global de~Sitter space are unique in the domain $D'':=\{\tau\leq 1,\,|X|<2+\tau\}\supset D'$ when initial data are posed at $\tau=1$ in $|X|<3$. Thus, in order to analyze waves in $D$, it suffices to analyze waves in $(0,1]_\tau\times\TT^3_L$ for any fixed $L>3$; take $L=4$ for definiteness.

  We then Taylor expand the functions $u_\pm(\tau,X)$ in~\eqref{EqdSGblAsy} at $(\tau,X)=(0,0)$. Consider the expansion of $\tau^{\lambda_\pm(\nu)}u_\pm$, which reads (writing `$\equiv$' for equality in Taylor series at $(\tau,X)=(0,0)$)
  \begin{align*}
    \tau^{\lambda_\pm(\nu)}u_\pm(\tau,X) &\equiv \tau^{\lambda_\pm(\nu)} \sum_{k,\alpha} \frac{1}{k!\alpha!}\pa_\tau^k \pa_X^\alpha u_\pm(0,0) \tau^k X^\alpha \\
      &\equiv \sum_{j=0}^\infty e^{-(\lambda_\pm(\nu)+j)t_*} \Biggl(\sum_{|\alpha|+2 q=j} \frac{1}{(2 q)!\alpha!}\pa_\tau^{2 q}\pa_X^\alpha u_\pm(0,0) x^\alpha\Biggr),
  \end{align*}
  since only even powers of $\tau$ appear in the Taylor expansion of $u_\pm$. From this expression, we can read off that every QNM has to be equal to $-i(\lambda_\pm(\nu)+j)$ for some $j\in\N_0$. We prove the converse for $d\notin 2\N_0$ and leave the case $d\in 2\N_0$ to the reader. Since $\pa_\tau^{2 q}\pa_X^\alpha u_\pm(0,0)=c_{q\alpha}(\nu)\Delta_X^q\pa_X^\alpha u_\pm(0,0)$ for certain nonzero constants $c_{q\alpha}(\nu)$ (computable by means of~\eqref{EqdSGblCons}), one can express the inner sum entirely in terms of the $j$-th order Taylor coefficients of $u_\pm(0,X)$ at $X=0$; it is then easy to see that for any $j\in\N_0$, one can choose $u_\pm(0,X)$ so that the coefficient of $e^{-(\lambda_\pm(\nu)+j)t_*}$ is a nonzero function of $x$. In fact, one can freely prescribe the coefficient of $x^\alpha$ for any $\alpha\in\N_0^3$ with $|\alpha|=j$; since the $L^2(\Sph^2)$-inner product of $x^\alpha|_{|x|=1}$ and $Y$ is equal to $0$ for all degree $l$ spherical harmonics $Y\in\CI(\Sph^2)$ if and only if $|\alpha|<l$ or $|\alpha|-l$ is odd, we also obtain the stated conclusion for $\QNM_{\dS,l}(\Lambda,\nu)$.

  The dimension of the space of mode solutions with angular momentum $l$ corresponding to $\omega\in\QNM_{\dS,l}(\Lambda,\nu)$ can be proven to be equal to the dimension $2 l+1$ of the space of degree $l$ spherical harmonics by a careful inspection of the above calculation as well. For an alternative and more conceptual argument, see Remark~\ref{RmkdSDim} below.
\end{proof}

\subsection{QNMs via meromorphic continuation}
\label{SsdSM}

We recall here the definition of QNMs as poles of the meromorphic continuation of a resolvent acting on suitably chosen function spaces. As a preparation for the arguments in~\S\ref{SC}, we define QNMs for the ODEs arising via separation of variables on the static model of de~Sitter space. This separation creates an artificial (polar coordinate) singularity at $r=0$ which will precisely match the conic singularity arising from gluing a small Schwarzschild black hole into de~Sitter space at $r=0$ in~\S\ref{SC}; moreover, this gives us the opportunity to relate calculations in the physics literature \cite{BradyChambersLaarakeersPoissonSdSFalloff} based on separation into spherical harmonics to the definition of mode solutions which are required to be smooth across $x=r\omega=0$. (Indeed, this polar coordinate singularity does not feature in the existing literature \cite{VasyMicroKerrdS,WarnickQNMs}.) Thus, working with the form~\eqref{EqIdSStaticExt} of the de~Sitter metric and applying the Klein--Gordon operator
\[
  \Box_{g_\dS}-\nu = -\pa_{t_*}^2 - (2 r\pa_r+3)\pa_{t_*} + r^{-2}\pa_r r^2(1-r^2)\pa_r + r^{-2}\Delta_{\Sph^2} - \nu
\]
to separated functions of the form $u(r)e^{-i\omega t_*}Y_{l m}$ produces the ordinary differential operator
\begin{equation}
\label{EqdSMOp}
  P_{l,0,\nu}(\omega) = r^{-2}\pa_r r^2(1-r^2)\pa_r + i\omega(2 r\pa_r+3) - r^{-2}l(l+1) + \omega^2-\nu.
\end{equation}
We define the function spaces on which we shall consider the action of this operator:

\begin{definition}[Weighted b-Sobolev spaces]
\label{DefdSM}
  Let $\alpha\in\R$, $k\in\N_0$, $r_0>0$. Then we define the weighted $L^2$- and b-Sobolev spaces
  \begin{align*}
    \Hb^{0,\alpha}([0,r_0)) &:= \bigl\{ u \colon r^{-\alpha}u \in L^2([0,r_0)_r;r^2\,\dd r) \bigr\}, \\
    \Hb^{k,\alpha}([0,r_0)) &:= \bigl\{ u \in \Hb^{0,\alpha}([0,r_0)) \colon (r\pa_r)^j u\in\Hb^{0,\alpha}([0,r_0))\ \forall\,j=0,\ldots,k \bigr\}.
  \end{align*}
\end{definition}

\begin{prop}[Fredholm property of $P_{l,0,\nu}(\omega)$]
\label{PropdSFred}
  Let $k\geq 2$, $\alpha\in(\half-l,\tfrac32+l)$. Let
  \[
    \cX^{k,\alpha} := \bigl\{ u\in\Hb^{k,\alpha}([0,2)) \colon P_{l,0,\nu}(0)u\in\Hb^{k-1,\alpha-2}([0,2)) \bigr\}.
  \]
  Then for $\Im\omega>\half-k$, the operator $P_{l,0,\nu}(\omega)\colon\cX^{k,\alpha}\to\Hb^{k-1,\alpha-2}([0,2))$ is Fredholm of index $0$ and invertible when $\Im\omega\gg 1$.
\end{prop}

The quasinormal modes of $\Box_{g_\dS}-\nu$ can now be defined in $\Im\omega>\half-k$ as the set of poles of $P_{l,0,\nu}(\omega)^{-1}\colon\Hb^{k-1,\alpha-2}\to\Hb^{k,\alpha}$; this set is discrete by the analytic Fredholm theorem. Moreover, the proof below shows that the intersection of this set with $\{\Im\omega>\half-k_0\}$ is independent of the choice of $k\geq k_0$, hence taking the union over all $k$ recovers the full set of QNMs in Proposition~\ref{PropdSGbl} in view of the smoothness of $e^{-i\omega t_*}u(r)Y_{l m}(\omega)$ on $\R_{t_*}\times\{ x=r\omega,\ |x|<2 \}$ for $u\in\ker P_{l,0,\nu}(\omega)$ (which we shall prove as well).

\begin{proof}[Proof of Proposition~\usref{PropdSFred}]
  Let us abbreviate
  \[
    P(\omega):=P_{l,0,\nu}(\omega).
  \]
  We give a proof of the Proposition in the spirit of \cite{VasyMicroKerrdS}, but, as in \cite{WarnickQNMs}, without microlocal techniques. We use b-analysis at the (regular singular) point $r=0$. In the present setting, one can certainly proceed even more explicitly than we do, but the perspective on a priori estimates here generalizes well to the analysis of the singular limit considered in~\S\ref{SC}, and also to settings in which separation of variables is unavailable or not desirable (cf.\ Remark~\ref{RmkIStronger}).

  \pfstep{Estimate for $P(\omega)$.} In $r<1$, we can write $r^2 P(\omega)=(r\pa_r)^2+a_1(r)r\pa_r+a_0(r)$ with smooth coefficients $a_0,a_1\in\CI([0,1)_r)$ satisfying $a_1(0)=1$ and $a_0(0)=-l(l+1)$. Thus, for any $r_0\in(0,1)$,
  \[
    \|u\|_{\Hb^{k,\alpha}([0,r_0))} \leq C\bigl(\|P(\omega)u\|_{\Hb^{k-2,\alpha-2}([0,r_0))} + \|u\|_{\Hb^{k-1,\alpha}([0,r_0))}\bigr).
  \]
  Fix a cutoff $\chi\in\CIc([0,r_0))$ which is identically $1$ on $[0,\frac{r_0}{2}]$. The interpolation inequality $\|u\|_{\Hb^{k-1,\alpha}}\leq C_\eps\|u\|_{\Hb^{0,\alpha}}+\eps\|u\|_{\Hb^{k,\alpha}}$, valid for any $\eps>0$, then gives for $\eps=1/(2 C)$
  \begin{equation}
  \label{EqdSFred0}
    \|u\|_{\Hb^{k,\alpha}([0,r_0))} \lesssim \|P(\omega)u\|_{\Hb^{k-2,\alpha-2}([0,r_0))} + \|\chi r^{-\alpha+\frac32}u\|_{L^2([0,r_0);\frac{\dd r}{r})} + \|(1-\chi)u\|_{\Hb^{0,\alpha}([0,r_0))}.
  \end{equation}
  (Here and below, we write `$A\lesssim B$' when there exists a constant $C$, independent of $u$, so that $A\leq C B$.)

  Next, let
  \begin{equation}
  \label{EqdSFredQ}
    Q:=(r\pa_r)^2+a_1(0)r\pa_r+a_0(0)=(r\pa_r)^2+r\pa_r-l(l+1);
  \end{equation}
  in the terminology of \cite{MelroseAPS}, this is the normal operator of $r^2 P(\omega)$ at $r=0$. Since $\alpha-\frac32\in(-l-1,l)$ lies between the indicial roots of $Q$ at $r=0$, use of the Mellin transform and Plancherel's theorem gives $\|r^{-\alpha+\frac32}\chi u\|_{L^2([0,r_0);\frac{\dd r}{r})}\lesssim \| r^{-\alpha+\frac32}Q(\chi u)\|_{L^2([0,r_0);\frac{\dd r}{r})}$. Since $P(\omega)-r^{-2}Q$ is bounded as a map $\Hb^{2,\alpha}\to\Hb^{1,\alpha-1}$, and since $[P(\omega),\chi]\colon\Hb^{1,N}\to\Hb^{0,N'}$ for any $N,N'$ (which we use for $N=\alpha-1$, $N'=\alpha$), we obtain
  \begin{align*}
    \|\chi r^{-\alpha+\frac32}u\|_{L^2([0,r_0);\frac{\dd r}{r})} &\lesssim \|P(\omega)(\chi u)\|_{\Hb^{0,\alpha-2}([0,r_0))} + \|\chi u\|_{\Hb^{1,\alpha-1}([0,r_0))} \\
      &\lesssim \|\chi P(\omega)u\|_{\Hb^{0,\alpha-2}([0,r_0))} + \|u\|_{\Hb^{1,\alpha-1}([0,r_0))}.
  \end{align*}
  Plugging this into~\eqref{EqdSFred0}, and using that $(1-\chi)\cdot(-)\colon\Hb^{0,\alpha}\to\Hb^{0,N}$ for any $N$, gives
  \begin{equation}
  \label{EqdSFred1}
    \|u\|_{\Hb^{k,\alpha}([0,r_0))} \leq C\bigl( \|P(\omega)u\|_{\Hb^{k-2,\alpha-2}([0,r_0))} + \|u\|_{\Hb^{1,\alpha-1}([0,r_0))} \bigr).
  \end{equation}

  We next estimate $u$ in $r\geq r_0$. Working in $r\geq r_0$ until further notice, hence with standard Sobolev spaces $H^k([r_0,2])$, we note that the coefficient of $\pa_r^2$ in $P(\omega)$ is nonzero for $r\neq 1$. Therefore, for any $r_1\in(1,2)$,
  \begin{equation}
  \label{EqdSFred2}
    \|u\|_{H^k([r_1,2])} \leq C\bigl( \|P(\omega)u\|_{H^{k-2}([r_1,2])} + \|u\|_{H^{k-1}([r_1,2])} \bigr).
  \end{equation}

  It remains to analyze $u$ in the vicinity $[r_0,r_1]$ of the cosmological horizon at $r=1$, where $P(\omega)$ degenerates; this will be an ODE version of the celebrated red-shift estimate \cite{DafermosRodnianskiRedShift}, or the microlocal radial point estimate \cite[\S2]{VasyMicroKerrdS}, \cite[\S5.4.2]{DyatlovZworskiBook}. Thus, let $r_0^-<r_0<1<r_1<r_1^+$, and let $\phi\in\CIc((r_0^-,r_1^+))$ be identically $1$ near $[r_0,r_1]$, and with $\phi'\leq 0$ on $[1,r_1^+)$. Let $\pa_r^*=-r^{-2}\pa_r r^2$ denote the adjoint of $\pa_r$ with respect to the $L^2((r_0^-,r_1^+);r^2\,\dd r)$ inner product. We then compute for $A=\phi(r)^2\pa_r^{2 k-1}-(\pa_r^*)^{2 k-1}\phi(r)^2=-A^*$ the inner product
  \begin{subequations}
  \begin{equation}
  \label{EqdSFredComm}
    2 \Re\la P(\omega)u,A u\ra = \la \sC u,u\ra,\quad \sC:=[P(\omega),A]+(P(\omega)^*-P(\omega))A.
  \end{equation}
  Here, $\sC$ is a differential operator of order $2 k$ whose coefficients have compact support in $(r_0^-,r_1^+)$; we proceed to compute its principal part. The principal part of $A$ is $2\phi(r)^2\pa_r^{2 k-1}$, that of $P(\omega)^*-P(\omega)$ is $2(-i)\bar\omega r(-\pa_r)-2 i\omega r\pa_r=4 r(\Im\omega)\pa_r$, and for the calculation of the principal part of $[P(\omega),A]$, we can replace $P(\omega)$, resp.\ $A$ by $\mu_0\pa_r^2$ where $\mu_0=1-r^2$, resp.\ $2\phi(r)^2\pa_r^{2 k-1}$. One then finds
  \begin{equation}
  \label{EqdSFredComm1}
    (-1)^k\sC = \bigl( (\pa_r^*)^k\, \phi(-2(2 k-1)\mu_0'(r) + 8 r\Im\omega)\phi\,\pa_r^k \bigr) + (\pa_r^*)^k\,4(1-r^2)\phi\phi'\,\pa_r^k + R,
  \end{equation}
  where $R$ is a $(2 k-1)$-th order differential operator with coefficients supported in $(r_0^-,r_1^+)$. Since $\mu_0'(r)<0$ near $r=1$ and
  \begin{equation}
  \label{EqdSFredCommPos}
    \eta(r):=(2 k-1)|\half\mu_0'(r)|+2 r\Im\omega
  \end{equation}
  satisfies $\eta(1)=2(k-\half+\Im\omega)>0$, the first summand in~\eqref{EqdSFredComm1}---the main term, as it is the only principal term which does not vanish at $r=1$---has a smooth square root on $[r_0^-,r_1^+]$ which is positive on $[r_0,r_1]$, provided we choose $r_0^-<1<r_1^+$ sufficiently close to $1$. In the second summand, $(1-r^2)\phi'$ is nonnegative both in $r<1$ and in $r>1$, i.e.\ has the same sign as the main term. Therefore, for any $\eps>0$ we have
  \begin{equation}
  \label{EqdSFredComm2}
    |\la\sC u,u\ra| \geq \|\eta^{\frac12}\phi\pa_r^k u\|_{L^2([r_0^-,r_1^+])}^2 - \eps\|u\|_{H^k([r_0^-,r_1^+])}^2 - C_\eps\|u\|_{H^{k-1}([r_0^-,r_1^+])}^2,
  \end{equation}
  the last two terms bounding $-|\la R u,u\ra|$ from below.

  The left hand side of~\eqref{EqdSFredComm}, in which we can integrate by parts $k-1$ times, is bounded in absolute value from above by
  \begin{equation}
  \label{EqdSFredComm3}
    4\delta^{-1}\|\phi\pa_r^{k-1}P(\omega)u\|_{L^2([r_0^-,r_1^+])}^2 + \|\delta^{\frac12}\phi\pa_r^k u\|_{L^2([r_0^-,r_1^+])}^2 + C_\delta\|u\|_{H^{k-1}([r_0^-,r_1^+])}^2
  \end{equation}
  \end{subequations}
  for any $\delta>0$. Fixing $\delta>0$ so small that $\eta^{\frac12}-\delta^{\frac12}\geq\half\eta^{\frac12}$ on $\supp\phi$, we thus obtain from~\eqref{EqdSFredComm}--\eqref{EqdSFredComm3}, and using that $\phi\equiv 1$ on $[r_0,r_1]$,
  \[
    \|u\|_{H^k([r_0,r_1])} \leq C_\eps\bigl(\|P(\omega)u\|_{H^{k-1}([r_0^-,r_1^+])} + \|u\|_{H^{k-1}([r_0^-,r_1^+])}\bigr) + \eps\|u\|_{H^k([r_0^-,r_1^+])}.
  \]
  Combining this with~\eqref{EqdSFred1} and \eqref{EqdSFred2} and choosing $\eps>0$ sufficiently small, we finally obtain
  \begin{equation}
  \label{EqdSFredFinal}
    \|u\|_{\Hb^{k,\alpha}([0,2))} \leq C\bigl( \|P(\omega)u\|_{\Hb^{k-1,\alpha-2}([0,2))} + \|u\|_{\Hb^{k-1,\alpha-1}([0,2))} \bigr).
  \end{equation}

  \pfstep{Fredholm property.} Since the inclusion $\Hb^{k,\alpha}([0,2))\hra\Hb^{k-1,\alpha-1}([0,2))$ is compact, this estimate implies that $P(\omega)\colon\cX^{k,\alpha}\to\Hb^{k-1,\alpha-2}$ has closed range and finite-dimensional kernel. Its range is the orthocomplement (with respect to $L^2([0,2);r^2\,\dd r)$) of the kernel of $P(\omega)^*$ acting on
  \[
    \bigl(\Hb^{k-1,\alpha-2}([0,2))\bigr)^*=\Hb^{-k+1,-\alpha+2}([0,2])^\bullet:=\bigl\{ v\in\Hb^{-k+1,-\alpha+2}([0,3)) \colon \supp v\subset[0,2] \bigr\}.
  \]
  Thus, if $P(\omega)^*v=0$ for $v\in\Hb^{-k+1,-\alpha+2}([0,2])^\bullet$, then $v=0$ in $(1,2]$ (by uniqueness of solutions of ODEs), so $\supp v\subset[0,1]$. At this point, one can either prove an estimate dual to~\eqref{EqdSFredFinal} for $P(\omega)^*$ (which would require the use of pseudodifferential techniques due to working with negative order spaces), or we may simply appeal to the nature of $P(\omega)^*$ as a nondegenerate second order differential operator on $(0,1)$, which implies that $v|_{(0,1)}$ must lie in a 2-dimensional space. Finally then, the kernel of the map $\Hb^{-k+1,-\alpha+2}([0,2])^\bullet\cap\ker P(\omega)^*\ni v\mapsto v|_{(0,1)}$ consists of distributions supported at $\{1\}$, thus of linear combinations of at most $(k-2)$-fold differentiated $\delta$-distributions. This proves that $\dim\ker P(\omega)^*\leq k$ is finite, and thus $P(\omega)$ is Fredholm.

  \pfstep{Triviality of the kernel for $\Im\omega\gg 1$.} Suppose $u\in\cX^{k,\alpha}\cap\ker P(\omega)$. We first claim that
  \[
    \tilde u(r\omega):=u(r)Y_{l m}(\omega)
  \]
  is a smooth function on $\{|x|<2\}$. Indeed, the above arguments imply that $u\in\Hb^{\infty,\alpha}([0,r_0))$ for $r_0<1$, and $u$ is smooth away from $r=0,1$. Now, working near $r=0$, inversion of the model operator $Q$ from~\eqref{EqdSFredQ}, and using that the weight $\alpha$ is such that $r^{-l-1}\notin\Hb^{\infty,\alpha}([0,\half))$, we conclude that $u=r^l u_0+\cO(r^{l+1})$ for some $u_0\in\C$. Since in fact $P(\omega)-r^{-2}Q=b_2(r)(r\pa_r)^2+b_1(r)r\pa_r+b_0(r)$ where $b_0,b_1,b_2$ are \emph{even} functions of $r$ (in fact, they are polynomials in $r^2$), we have $u=r^l u_0(r)$, where $u_0$ is even. Therefore, $v_0(x):=u_0(|x|)$ is a smooth function of $|x|^2$, hence smooth near $x=0$. Thus, $\tilde u(x)=r^l Y_{l m}(\omega) v_0(x)$ is smooth near $x=0$ since $r^l Y_{l m}(\omega)$ is.

  At $r=1$, one notes that
  \begin{equation}
  \label{EqdSFredRegSing}
    (1-r^2)P(\omega)u=0
  \end{equation}
  is a regular-singular ODE at $r=1$ with indicial roots $0$ and $-\Im\omega$ (see e.g.\ \cite[Appendix~A]{CardosoCostaDestounisHintzJansenSCC2}). For $u\in\cX^{k,\alpha}\subset\Hb^{k,\alpha}$ with $k>\half-\Im\omega$, the asymptotics arising from the root $-\Im\omega$ cannot occur, hence $u$ is smooth across $r=1$.

  We conclude that $U(t_*,x)=e^{-i\omega t_*}\tilde u(x)$ is a smooth solution of $(\Box_{g_\dS}-\nu)U=0$ on the domain $D=[0,\infty)_{t_*}\times\{|x|<2\}$ containing the static patch of de~Sitter space. An energy estimate shows that there exists a constant $C$ (depending only on $\nu$) such that for any smooth solution $U$, one has $|U(t_*,x)|\leq C' e^{C t_*}$, with $C'$ depending on the Cauchy data of $U$ at $t_*=0$. But this implies that $\Im\omega\leq C$ or $\tilde u=0$. Therefore, if $\Im\omega>C$, then necessarily $\tilde u=0$ and thus $u=0$, hence the kernel of $P(\omega)$ is trivial.

  \pfstep{Triviality of the cokernel; index zero.} Consider $v\in\Hb^{-k+1,-\alpha+2}([0,2])^\bullet$ satisfying $P(\omega)^*v=0$. We have $P(\omega)^*=P_{0,\bar\nu}(\bar\omega)$. Since $-\alpha+2\in(\half-l,\tfrac32+l)$ lies in the same interval as $\alpha$, the previous arguments thus imply that $\tilde v(x):=v(r)Y_{l m}(\omega)$ is smooth in $|x|<1$; in $|x|>1$, $\tilde v$ vanishes identically. Near $r=1$ on the other hand, the indicial roots of $(1-r^2)P(\omega)^*$ are $0$ and $-\Im\bar\omega=\Im\omega$. Let us restrict attention to the case $\Im\omega>0$. Since $\supp v\subset\{r\leq 1\}$, but $P(\omega)H(1-r)=-2 i\omega\delta(r-1)$ plus smoother terms, and since $P(\omega)((1-r)^j H(1-r))$ is continuous for $j\geq 1$ (thus cannot cancel the $\delta(r-1)$ distribution), the only contribution to $v$ near $r=1$ comes from the indicial root $\Im\omega$; hence, $v\in H^{\frac12+\Im\omega-\eps}([\half,\tfrac32])$ for any $\eps>0$. For $\Im\omega\gg 1$ then, $V(t_*,x):=e^{-i\bar\omega t_*}\tilde v(x)$ is thus a finite energy solution of the wave equation on de~Sitter space with fast exponential decay as $t_*\to\infty$; an energy estimate shows that this implies $V\equiv 0$ and thus $v\equiv 0$ when $\Im\omega$ is sufficiently large.

  We conclude that $P(\omega)\colon\cX^{k,\alpha}\to\Hb^{k-1,\alpha-2}$ is invertible, in particular has Fredholm index $0$, for $\Im\omega$ sufficiently large. Since the index is independent of $\omega$ with $\Im\omega>\half-k$, the proof is complete.
\end{proof}

\begin{rmk}[The smooth kernel of $P(\omega)$ is at most 1-dimensional]
\label{RmkdSDim}
  Due to the regular-singular nature of $(1-r^2)P_{0,\mu}(\omega)$ at $r=1$, any element $u\in\ker P(\omega)\cap\Hb^{k,\alpha}([0,2))\subset\CI((0,2))$ is uniquely determined by its restriction to the interval $(0,1)$ on which $P(\omega)$ is a regular second order operator. Corresponding to the indicial root $-l-1$ at the singular point $r=0$, $P(\omega)v=0$ admits a solution $v(r)=r^{-l-1}+o(r^{-l-1})$ on $(0,1)$ which does not lie in $\Hb^{0,\alpha}([0,2))$ and is thus linearly independent from $u$. This implies that $\ker P(\omega)$ is at most $1$-dimensional.
\end{rmk}

\section{Convergence of QNMs}
\label{SC}

Having determined the QNMs of de~Sitter space, we now turn to the convergence claim of Theorem~\ref{ThmIMain}.

\subsection{Schwarzschild--de~Sitter black hole spacetimes with mass tending to zero}
\label{SsCM}

We first need to pass from the static coordinates used in~\eqref{EqISdS} to coordinates valid across the event and cosmological horizons, in a manner that is uniform down to $\bhm=0$ in a suitable sense. Concretely, we introduce a function
\begin{equation}
\label{EqCMTimeFn}
  t_{*,\bhm}=t-F_\bhm(r)
\end{equation}
whose level sets are transversal to the future event and cosmological horizon, and where $F_\bhm$ is nonzero only for $r/\bhm\sim 1$ (near the event horizon) and for $r\sim 1$ (near the cosmological horizon---recall that $\Lambda=3$). To motivate the precise choice of $F_\bhm$, we compute
\begin{align}
  \Box_{g_\bhm}-\nu &= -\mu_\bhm^{-1}\pa_t^2 + r^{-2}\pa_r r^2\mu_\bhm\pa_r + r^{-2}\Delta_{\Sph^2} - \nu \nonumber\\
\label{EqCMOp}
\begin{split}
    &= \mu_\bhm^{-1}\bigl((\mu_\bhm F_\bhm')^2-1\bigr)\pa_{t_{*,\bhm}}^2 - (r^{-2}\pa_r r^2\mu_\bhm F_\bhm'+\mu_\bhm F_\bhm'\pa_r)\pa_{t_{*,\bhm}} \\
    &\qquad + r^{-2}\pa_r r^2\mu_\bhm\pa_r + r^{-2}\Delta_{\Sph^2}-\nu.
\end{split}
\end{align}

The singularities of $\mu_\bhm(r)^{-1}$ are controlled by the following lemma:
\begin{lemma}[Location of the horizons]
\label{LemmaCMRoots}
  The positive roots $0<r_+(\bhm)<r_\cC(\bhm)$ of $\mu_\bhm$ depend smoothly on $\bhm\in[0,\frac{1}{3\sqrt 3})$. They moreover satisfy
  \[
    \lim_{\bhm\searrow 0} \frac{r_+(\bhm)}{\bhm} = 2,\qquad
    \lim_{\bhm\searrow 0} r_\cC(\bhm) = 1.
  \]
\end{lemma}
\begin{proof}
  We only discuss the case of very small $\bhm>0$. The number $r_\cC(\bhm)$ is the largest root of the polynomial $-r\mu_\bhm=r^3-r+2\bhm$; since the largest root $r_\cC(0)=1$ of $-r\mu_0=r^3-r$ is simple, $r_\cC(\bhm)$ is an analytic function of $\bhm$ near $\bhm=0$. For the smallest positive root $r_+(\bhm)$, we write $r_+(\bhm)=\bhm\hat r_+(\bhm)$, with $\hat r_+(\bhm)$ being the smallest positive solution of $0=1-\frac{2}{\hat r}-\bhm^2\hat r^2$; the root $\hat r_+(0)=2$ is simple, which implies the result.
\end{proof}

Fix now $\chi_+\in\CIc([0,4))$ to be identically $1$ on $[0,3]$ and $\chi_\cC\in\CIc((\tfrac14,4))$ to be identically $1$ on $[\half,2]$. We then define
\begin{equation}
\label{EqCMTimeFnMod}
  F_\bhm(r) := \chi_+(r/\bhm)\int_r^{4\bhm} \mu_\bhm^{-1}(r')\,\dd r' + \chi_\cC(r)\int_{1/4}^r \mu_\bhm^{-1}(r')\,\dd r',\quad r\in(r_+(\bhm),r_\cC(\bhm)).
\end{equation}
In particular, for $0<\bhm<\tfrac{1}{16}$, the smooth function
\begin{equation}
\label{EqCMalpha}
  \alpha_\bhm := \mu_\bhm F_\bhm'
\end{equation}
is equal to $-1$ for $r<3 r_+(\bhm)$ and equal to $+1$ for $r>\half r_\cC(\bhm)$. Thus, in these two intervals,
\[
  \Box_{g_\bhm}-\nu = r^{-2}\pa_r r^2\mu_\bhm\pa_r \pm 2 r^{-1}\pa_r r\pa_{t_{*,\bhm}} + r^{-2}\Delta_{\Sph^2} - \nu.
\]
(The top sign is for $r<3 r_+(\bhm)$, the bottom sign for $r>\half r_\cC(\bhm)$.) The coefficients are thus analytic in these regions and can be analytically continued across $r=r_+(\bhm)$ and $r=r_\cC(\bhm)$. We shall regard $\Box_{g_\bhm}-\nu$ as the operator on the manifold $\R_{t_{*,\bhm}}\times[\half r_+(\bhm),2 r_\cC(\bhm)]\times\Sph^2$ defined by~\eqref{EqCMOp} and its analytic continuation. The action on the radial part $u(r)$ of separated functions of the form $e^{-i\omega t_{*,\bhm}}Y_{l m}(\omega)u(r)$ is given by
\begin{equation}
\label{EqCMP}
\begin{split}
  P_{l,\bhm,\nu}(\omega) &= r^{-2}\pa_r r^2\mu_\bhm\pa_r - r^{-2}l(l+1) - \nu \\
    &\qquad + i \omega(r^{-2}\pa_r r^2\alpha_\bhm+\alpha_\bhm\pa_r) + \beta_\bhm\omega^2, \quad
  \beta_\bhm:=\mu_\bhm^{-1}(1-\alpha_\bhm^2).
\end{split}
\end{equation}

\subsection{QNMs of Schwarzschild--de~Sitter spacetimes}
\label{SsCSdS}

For fixed $\bhm\in(0,\frac{1}{3\sqrt{3}})$, we can analyze the operator $P_{l,\bhm,\nu}(\omega)$ on the interval
\[
  I_\bhm:=\bigl(\half r_+(\bhm),2 r_\cC(\bhm)\bigr)
\]
similarly to Proposition~\ref{PropdSFred}; one now has two horizons at which one performs a positive commutator estimate as in~\eqref{EqdSFredComm}, but one can work with standard Sobolev spaces as we are working away from the polar coordinate singularity at $r=0$. The upshot is:

\begin{prop}[QNMs of SdS]
\label{PropCSdS}
  A number $\omega\in\C$ is a QNM of $\Box_{g_\bhm}-\nu$ with an angular momentum $l\in\N_0$ resonant state if and only if the kernel of $P_{l,\bhm,\nu}(\omega)$ on $H^k(I_\bhm)$ is trivial for any $k$ with $k\geq 2$ and $k>\half-\min(\Im\omega/\kappa_+(\bhm),\Im\omega/\kappa_\cC(\bhm))$.
\end{prop}

See~\eqref{EqdSFredCommPos} (with $\mu_\bhm$ instead of $\mu_0$) regarding the regularity requirement. We also note
\[
  \kappa_+(\bhm)=|\half\mu_\bhm'(r_+(\bhm))| = \frac{1}{4\bhm} + \cO(1),\quad
  \kappa_\cC(\bhm)=|\half\mu_\bhm'(r_\cC(\bhm))| = 1 + \cO(\bhm).
\]
Thus, for fixed $k\geq 2$ and $C\in\R$, and for $\omega\in\C$ with $\Im\omega\geq C$, and under the assumption $k>\half-C$ as in the de~Sitter case in Proposition~\ref{PropdSFred}, the regularity requirement in Proposition~\ref{PropCSdS} is satisfied for all such $\omega$ provided $\bhm>0$ is sufficiently small.

We remark that resonant states are automatically analytic if instead of $t_{*,\bhm}$ we work with an analytic time function, see \cite{PetersenVasyAnalytic,GalkowskiZworskiHypo}.

\subsection{Proof of the main theorem}
\label{SsCPf}

As a consequence of~\eqref{EqIResc}, in the regime $r\sim\bhm$ one should measure regularity of functions of $r$ by means of differentiation along $\pa_{\hat r}=\bhm\pa_r\sim r\pa_r$; for $r\sim 1$ on the other hand, one uses $\pa_r$ (as in the part of the proof of Proposition~\ref{PropdSFred} taking place away from $r=0$), which is also comparable to $r\pa_r$ there. Thus, measuring regularity with respect to $r\pa_r$ captures both regimes simultaneously. We shall thus estimate $P_{l,\bhm,\nu}(\omega)$ using $\bhm$-dependent function spaces
\begin{align*}
  &\Hb^{k,\alpha}(I_\bhm),\quad \|u\|_{\Hb^{k,\alpha}(I_\bhm)}^2 := \sum_{j\leq k} \|r^{-\alpha} (r\pa_r)^j u(r)\|_{L^2(I_\bhm;r^2\,\dd r)}^2.
\end{align*}
Thus, for fixed $\bhm>0$, we have $\Hb^{k,\alpha}(I_\bhm)\cong H^k(I_\bhm)$, but the norms of the two spaces are not uniformly equivalent as $\bhm\searrow 0$. We begin with a uniform symbolic estimate:

\begin{lemma}[Estimate of the highest derivative]
\label{LemmaCSymb}
  Let $k\geq 2$, $\alpha\in\R$, and let $K\subset\C$ be compact. Suppose $\min_{\omega\in K}\Im\omega>\half-k$. Then there exist $C>0$ and $\bhm_0>0$ so that for $\bhm<\bhm_0$ and all $\omega\in K$,
  \begin{equation}
  \label{EqCSymbEst}
    \|u\|_{\Hb^{k,\alpha}(I_\bhm)} \leq C\bigl(\|P_{l,\bhm,\nu}(\omega)u\|_{\Hb^{k-1,\alpha-2}(I_\bhm)} + \|u\|_{\Hb^{0,\alpha}(I_\bhm)}\bigr).
  \end{equation}
\end{lemma}
\begin{proof}
  We first work away from the horizons, where we have elliptic estimates. Indeed,
  \begin{align*}
    r^2 P_{l,\bhm,\nu} &= \mu_\bhm(r\pa_r)^2 + \bigl((r\mu_\bhm)'+2 i\omega\alpha_\bhm r\bigr)r\pa_r + \bigl(-l(l+1) - r^2\nu + i\omega(r^2\alpha_\bhm)' + r^2\beta_\bhm\omega^2\bigr) \\
      &=: \mu_\bhm(r\pa_r)^2 + a_\bhm(\omega,r)r\pa_r + b_\bhm(\omega,r).
  \end{align*}
  where the coefficients obey uniform bounds
  \[
    |(r\pa_r)^j\mu_\bhm|,\ |(r\pa_r)^j a_\bhm|,\ |(r\pa_r)^j b_\bhm| \leq C_j\quad\text{on}\ I_\bhm\quad \forall\,j\in\N_0.
  \]
  Fixing $0<\delta\leq\frac{1}{10}$ and letting
  \[
    I_{\bhm,\bullet,\delta}:=\bigl((1-\delta)r_\bullet(\bhm),(1+\delta)r_\bullet(\bhm)\bigr),\quad \bullet=+,\cC,\qquad
    J_{\bhm,\delta} := I_\bhm\setminus(I_{\bhm,+,\delta}\cup I_{\bhm,\cC,\delta}),
  \]
  one moreover has a uniform lower bound $|\mu_\bhm(r)|\geq c_\delta>0$ on $J_{\bhm,\delta}$, thus $|(r\pa_r)^j\mu_\bhm^{-1}|\leq C_{\delta,j}$ on $J_{\bhm,\delta}$ for all $j\in\N_0$. This implies the uniform estimate
  \begin{equation}
  \label{EqCSymbEll}
    \|u\|_{\Hb^{k,\alpha}(J_{\bhm,\delta})} \leq C_\delta\bigl(\|P_{l,\bhm,\nu}(\omega)u\|_{\Hb^{k-2,\alpha-2}(J_{\bhm,\delta})} + \|u\|_{\Hb^{k-1,\alpha}(J_{\bhm,\delta})} \bigr).
  \end{equation}

  Next, we work near the cosmological horizon, more precisely on $I_{\bhm,\cC,2\delta}$; the operator $P_{l,\bhm,\nu}(\omega)$ degenerates at $r=r_\cC(\bhm)$. The positive commutator estimate based on~\eqref{EqdSFredComm} and the subsequent calculations apply here as well, with the same commutant $A$. Indeed, the calculations and estimates~\eqref{EqdSFredComm}--\eqref{EqdSFredComm3} apply for $\bhm=0$, and are based on the nondegeneracy of the first term in~\eqref{EqdSFredComm1}; the positivity guaranteeing this is an open condition in the coefficients of the differential operator. Thus, we obtain (the weight $\alpha$ being irrelevant here, as we are working in $r\gtrsim 1$) for any $\eps>0$ and $N\in\R$,
  \begin{equation}
  \label{EqCSymbCosm}
    \|u\|_{\Hb^{k,\alpha}(I_{\bhm,\cC,\delta})} \leq C_{\delta,\eps}\bigl(\|P_{l,\bhm,\nu}(\omega)u\|_{\Hb^{k-1,-N}(I_{\bhm,\cC,2\delta})} + \|u\|_{\Hb^{k-1,-N}(I_{\bhm,\cC,2\delta})}\bigr) + \eps\|u\|_{\Hb^{k,\alpha}(I_{\bhm,\cC,2\delta})}.
  \end{equation}

  Near the event horizon, i.e.\ near $I_{\bhm,+,\delta}$, we use a similar red-shift estimate. Passing to the coordinate $\hat r=r/\bhm$, so $r\pa_r=\hat r\pa_{\hat r}$, note that in $I_{\bhm,+,\delta}$ (where $\alpha_\bhm=-1$ and $\beta_\bhm=0$), we have
  \begin{equation}
  \label{EqCSymbSchwResc}
    \bhm^2 P_{l,\bhm,\nu}(\omega) = \hat r^{-2}\pa_{\hat r}\hat r^2 \mu_\bhm(\bhm\hat r)\pa_{\hat r} - \hat r^{-2}l(l+1) - \bhm^2\nu - 2 i\bhm\omega\hat r^{-1}\pa_{\hat r}\hat r,
  \end{equation}
  with $\mu_\bhm(\bhm\hat r)=1-\frac{2}{\hat r}-\bhm^2\hat r^2\to \hat\mu(\hat r):=1-\frac{2}{\hat r}$. Upon taking the limit $\bhm\to 0$ in the coefficients of~\eqref{EqCSymbSchwResc}, we obtain the operator
  \begin{equation}
  \label{EqCSymbSchw}
    \hat P := \hat r^{-2}\pa_{\hat r}\hat r^2\hat\mu\pa_{\hat r} - \hat r^{-2}l(l+1),\qquad \hat\mu=1-\frac{2}{\hat r}.
  \end{equation}
  The principal part of $\hat P$ degenerates at $\hat r=2$. There, one can again run a positive commutator argument with commutant $\phi(\hat r)^2\pa_{\hat r}^{2 k-1}-(\pa_{\hat r}^*)^{2 k-1}\phi(\hat r)^2$ and exploit the fact that $\half|\pa_{\hat r}\hat\mu(2)|=\frac14\neq 0$. Here $\phi(\hat r)$ is supported in $\hat I_{2\delta}:=[2(1-2\delta),2(1+2\delta)]$ and equal to $1$ near $\hat I_\delta:=[2(1-\delta),2(1+\delta)]$, and $\pa_{\hat r}^*=-\hat r^{-2}\pa_{\hat r}\hat r^2$ is the adjoint with respect to $L^2(\hat I_{2\delta};\hat r^2\,\dd\hat r)$. Provided that $(2 k-1)\hat r>0$ at $\hat r=2$ (cf.\ \eqref{EqdSFredCommPos}), so $k>\half$, and taking $\delta>0$ sufficiently small, one can estimate
  \[
    \|u\|_{H^k(\hat I_\delta)} \leq C_\eps\bigl(\|\hat P u\|_{H^{k-1}(\hat I_{2\delta})} + \|u\|_{H^{k-1}(\hat I_{2\delta})}\bigr) + \eps\|u\|_{H^k(\hat I_{2\delta})}
  \]
  for any $\eps>0$. Using the openness of the positivity used in the proof of this estimate, one obtains (using the same commutant) the estimate
  \begin{equation}
  \label{EqCSymbEv}
    \|u\|_{\Hb^{k,\alpha}(I_{\bhm,+,\delta})} \leq C_\eps\bigl(\|P_{l,\bhm,\nu}(\omega)u\|_{\Hb^{k-1,\alpha-2}(I_{\bhm,+,2\delta})} + \|u\|_{\Hb^{k-1,\alpha}(I_{\bhm,+,2\delta})}\bigr) + \eps\|u\|_{\Hb^{k,\alpha}(I_{\bhm,+,2\delta})};
  \end{equation}
  regarding the weights, we use that on $I_{\bhm,+,2\delta}$, powers of $\bhm$ and $r$ are interchangeable.

  Adding the estimates~\eqref{EqCSymbEll}, \eqref{EqCSymbCosm}, and \eqref{EqCSymbEv}, and taking $\eps>0$ sufficiently small, one obtains the estimate~\eqref{EqCSymbEst}. (The differential order of the final, error, term can be shifted to $0$ using an interpolation inequality.)
\end{proof}

The error term in~\eqref{EqCSymbEst}, i.e.\ the second term on the right, is acceptable from the perspective of Fredholm theory for any individual value of $\bhm>0$, as the inclusion $\Hb^{k,\alpha}(I_\bhm)\to\Hb^{k-1,\alpha}(I_\bhm)$ is compact for any fixed $\bhm>0$ (the weight being irrelevant then). However, the error term is not \emph{small} as $\bhm\searrow 0$. We improve it in two stages, corresponding to the two model problems introduced in~\S\ref{SI}.

We begin by proving the invertibility of the model problem on the unit mass Schwarzschild spacetime. For $\hat r_0>0$, define the function spaces
\begin{align*}
  \Hb^{0,\gamma}((\hat r_0,\infty]) &:= \bigl\{ u \colon \hat r^\gamma u\in L^2((\hat r_0,\infty);\hat r^2\,\dd\hat r) \bigr\}, \\
  \Hb^{k,\gamma}((\hat r_0,\infty]) &:= \bigl\{ u\in\Hb^{0,\gamma}((\hat r_0,\infty]) \colon (\hat r\pa_{\hat r})^j u\in\Hb^{0,\gamma}((\hat r_0,\infty])\ \forall\,j=0,\ldots,k \bigr\}.
\end{align*}
Define the operator $\hat P=\hat r^{-2}\pa_{\hat r}\hat r^2\hat\mu\pa_{\hat r}-\hat r^{-2}l(l+1)$, $\hat\mu=1-\frac{2}{\hat r}$, as in~\eqref{EqCSymbSchw}.

\begin{lemma}[Invertibility of the zero energy wave operator on Schwarzschild]
\label{LemmaCSchw}
  Fix $\hat r_0\in(0,2)$. Let $k\geq 2$ and $\gamma\in(-\tfrac32-l,-\half+l)$. Put
  \[
    \hat\cX_{\hat r_0}^{k,\gamma} := \bigl\{ u\in\Hb^{k,\gamma}((\hat r_0,\infty]) \colon \hat P u\in\Hb^{k-1,\gamma+2}((\hat r_0,\infty]) \bigr\}.
  \]
  Then $\hat P\colon\hat\cX_{\hat r_0}^{k,\gamma}\to\Hb^{k-1,\gamma+2}((\hat r_0,\infty])$ is invertible. Moreover, for any compact $J\subset(0,2)$, the operator norm of $\hat P^{-1}\colon\Hb^{k-1,\gamma+2}((\hat r_0,\infty])\to\hat\cX_{\hat r_0}^{k,\gamma}$ is uniformly bounded for $\hat r_0\in J$.
\end{lemma}

The weight of the function spaces here controls the rate of decay as $\hat r\to\infty$, in contrast to the function spaces in Proposition~\ref{PropdSFred} where the rate of decay is measured as $r\to 0$. This gives meaning to the matching of a (rescaled) large end of a cone ($\hat r\to\infty$ on Schwarzschild) into the small end of a cone ($r\to 0$ on dS), cf.\ Figure~\ref{FigIBlowup}.

\begin{proof}[Proof of Lemma~\usref{LemmaCSchw}]
  We give a proof of this standard fact for completeness here; see also \cite[Theorem~6.1]{HaefnerHintzVasyKerr}. Elliptic and positive commutator estimates give (for any fixed $\gamma$)
  \[
    \|u\|_{\Hb^{k,\gamma}} \leq C\bigl(\|\hat P u\|_{\Hb^{k-1,\gamma+2}} + \|u\|_{\Hb^{0,\gamma}}\bigr).
  \]
  Fix a cutoff $\chi\in\CI(\R)$ which is identically $1$ for $\hat r\geq 4$ and vanishes for $\hat r\leq 3$. Denote by
  \[
    \hat Q = \pa_{\hat r}\hat r^2\pa_{\hat r} - l(l+1) = (\hat r\pa_{\hat r})^2 + \hat r\pa_{\hat r} - l(l+1)
  \]
  the model operator of $\hat r^2\hat P$ at $\hat r=\infty$; that is, the difference $\hat r^2\hat P-\hat Q=\hat a_1\hat r\pa_{\hat r}+\hat a_0$ has decaying coefficients, namely $|(\hat r\pa_{\hat r})^j\hat a_i|\leq C_j\hat r^{-1}$ for $i=0,1$ and all $j\in\N_0$. Using the Mellin transform in $\hat r^{-1}$ and using that $\gamma+\tfrac32\in(-l,l+1)$ (with $-l$, $l+1$ being the indicial roots of $\hat Q$ at $\hat r^{-1}=0$), we have
  \[
    \| \chi u \|_{\Hb^{0,\gamma}} = \| \hat r^{\gamma+3/2}\chi u \|_{L^2((\hat r_0,\infty);\frac{\dd\hat r}{\hat r})} \leq C \| \hat r^{\gamma+3/2}\hat Q(\chi u) \|_{L^2((\hat r_0,\infty));\frac{\dd\hat r}{\hat r})} = C\|\hat r^{-2}\hat Q(\chi u)\|_{\Hb^{0,\gamma+2}}.
  \]
  Arguing similarly to the proof of Proposition~\ref{PropdSFred}, replacing $\hat r^{-2}\hat Q$ by $\hat P$ and estimating the error and commutator (with $\chi$) terms gives
  \[
    \|u\|_{\Hb^{k,\gamma}} \leq C\bigl(\|\hat P u\|_{\Hb^{k-1,\gamma+2}} + \|u\|_{\Hb^{0,\gamma-1}}\bigr).
  \]
  This implies that $\hat P\colon\hat\cX_{\hat r_0}^{k,\gamma}\to\Hb^{k-1,\gamma+2}$ has closed range and finite-dimensional kernel. Directly by ODE theory, its cokernel is finite-dimensional.

  Now if $u\in\Hb^{k,\gamma}$ satisfies $\hat P=0$, then automatically $u\in\Hb^{\infty,\gamma'}$ for all $\gamma'<-\half+l$. Sobolev embedding (after passing to the coordinate $z=-\log\hat r$, in which b-Sobolev spaces are exponentially weighted standard Sobolev spaces on $(-\infty,-\log\hat r_0)$) implies that $|u|,|\hat r\pa_{\hat r}u|\lesssim\hat r^{-l-1+\eps}$ as $\hat r\to\infty$ for any $\eps>0$. This is sufficient for the boundary term at infinity in the integration by parts
  \begin{equation}
  \label{EqCSchwIBP}
    0 = -\int_2^\infty \hat P u \bar u\,\hat r^2\,\dd\hat r = \int_2^\infty \bigl(\hat r^2\hat\mu|\pa_{\hat r}u|^2 + l(l+1)|u|^2\bigr)\,\hat r^2\,\dd\hat r
  \end{equation}
  to vanish; the boundary term at $\hat r=2$ vanishes since $u$ is smooth there and $\hat\mu(2)=0$. Since $l\geq 0$, we see that $u$ is constant in $\hat r\geq 2$, and thus must vanish there since $\lim_{\hat r\to\infty}u=0$. Turning to the other side of the black hole event horizon, $u$ is now a smooth solution of the regular-singular ODE $(\hat r-2)\hat P u=0$ in $\hat r<2$ and vanishes to infinite order at $\hat r=2$. Thus, $u$ vanishes identically. This shows the injectivity of $\hat P$.

  To prove surjectivity, consider $v\in\Hb^{-k+1,-\gamma-2}([\hat r_0,\infty])^\bullet$, meaning that $v$ is a distribution in $\Hb^{-k+1,-\gamma-2}((\half\hat r_0,\infty])$ with support in $[\hat r_0,\infty)$. We need to show that if $\hat P^*v=0$, then $v=0$. We certainly have $v\equiv 0$ in $\hat r<2$ where $\hat P$ has non-vanishing leading order term. Next, since $-\gamma-2\in(-\tfrac32-l,-\half+l)$ lies in the same interval as $\gamma$ itself, we have $v\in\Hb^{\infty,\gamma'}((\hat r_1,\infty])$ for all $\hat r_1>2$ and all $\gamma'<-\half+l$.
  
  In $\hat r>2$, we solve the regular-singular ODE $(\hat r-2)\hat P^*v=0$; since $(\hat r-2)\hat P^*=(\hat r-2)\hat P$ has a double indicial root at $\hat r=2$, we obtain $v=H(\hat r-2)(v_0+v_1\log\hat\mu+\tilde v)$, where $v_0,v_1\in\C$, while $\tilde v$ and its derivatives along any power of $(\hat r-2)\pa_{\hat r}$ are bounded by $(\hat r-2)^{1-\eps}$ for $\hat r\in(2,3)$ and any $\eps>0$. One computes $\hat P^* H(\hat r-2)=0$, whereas
  \[
    \hat P^*\bigl(v_1 H(\hat r-2)\log\hat\mu\bigr) = \frac{v_1}{2\bhm}\delta(r-2\bhm)
  \]
  cannot be cancelled by $\hat P^*\tilde v$ unless $v_1=0$. But $v_1=0$ implies that $\hat\mu\pa_{\hat r}v\to 0$ as $\hat r\searrow 2$. This implies that the integration by parts as in~\eqref{EqCSchwIBP} does not produce a boundary term at $\hat r=2$; the boundary term at infinity is absent as before. We conclude that $v\equiv 0$, completing the proof of invertibility for any fixed $\hat r_0\in(0,2)$.

  Denote the inverse of $\hat P$ acting on distributions on $(\hat r_0,\infty)$ by $\hat P_{\hat r_0}^{-1}$. Then for any other $\hat r_1\in(0,2)$ with $\hat r_1\geq\hat r_0$, one can construct $\hat P_{\hat r_1}^{-1}f$ as follows: extend $f\in\Hb^{k-1,\gamma+2}((\hat r_1,\infty])$ to $\tilde f\in\Hb^{k-1,\gamma+2}((\hat r_0,\infty])$ with $\|\tilde f\|_{\Hb^{k-1,\gamma+2}((\hat r_0,\infty])}\leq C\|f\|_{\Hb^{k-1,\gamma+2}((\hat r_1,\infty])}$, where the constant $C$ only depends on $k,\hat r_0,\hat r_1$. Then $\hat P_{\hat r_1}^{-1}f=(\hat P_{\hat r_0}^{-1}\tilde f)|_{(\hat r_1,\infty)}$. This implies the final part of the Lemma.
\end{proof}

Consider now the error term $\|u\|_{\Hb^{0,\alpha}(I_\bhm)}$ of Lemma~\ref{LemmaCSymb}; we improve it in the regime $r\sim\bhm$ by inverting the Schwarzschild model problem. Fix a cutoff $\chi=\chi(r)\in\CIc([0,\half))$ which is identically $1$ on $[0,\frac14]$; then under the change of variables $r=\bhm\hat r$,
\begin{align*}
  \|\chi(r)u\|_{\Hb^{0,\alpha}(I_\bhm)} &= \|r^{-\alpha}\chi u\|_{L^2([\frac12 r_+(\bhm),\frac12);r^2\,\dd r)} = \bhm^{\frac32-\alpha}\|\hat r^{-\alpha}\chi(\bhm\hat r) u\|_{L^2([\frac{r_+(\bhm)}{2\bhm},\frac{1}{2\bhm});\hat r^2\,\dd\hat r)} \\
    &= \bhm^{\frac32-\alpha}\|\chi u\|_{\Hb^{0,-\alpha}((\hat r_+(\bhm),\infty])},\qquad
  \hat r_+(\bhm):=\frac{r_+(\bhm)}{2\bhm},
\end{align*}
similarly for higher order Sobolev spaces. By Lemmas~\ref{LemmaCMRoots} and~\ref{LemmaCSchw}, we have
\[
  \|\chi u\|_{\Hb^{1,-\alpha}((\hat r_+(\bhm),\infty])} \lesssim \|\hat P(\chi u)\|_{\Hb^{0,-\alpha+2}((\hat r_+(\bhm),\infty])},\qquad \alpha\in(\half-l,\tfrac32+l).
\]
Thus, estimating $\bhm^2 P_{l,\bhm,\nu}(\omega)-\hat P$ using \eqref{EqCSymbSchwResc}--\eqref{EqCSymbSchw},
\begin{align*}
  \|\chi u\|_{\Hb^{1,\alpha}(I_\bhm)} &\lesssim \bhm^{-2}\|\hat P(\chi u)\|_{\Hb^{0,\alpha-2}(I_\bhm)} \\
    &\lesssim \|P_{l,\bhm,\nu}(\omega)(\chi u)\|_{\Hb^{0,\alpha-2}(I_\bhm)} + \|\chi u\|_{\Hb^{2,\alpha-1}(I_\bhm)} \\
    &\lesssim \|\chi P_{l,\bhm,\nu}(\omega)u\|_{\Hb^{0,\alpha-2}(I_\bhm)} + \|u\|_{\Hb^{2,\alpha-1}(I_\bhm)},
\end{align*}
where the implicit constants are independent of $\bhm$; in the last step, we use that the commutator $[P_{l,\bhm,\nu}(\omega),\chi]\colon\Hb^{1,-N}(I_\bhm)\to\Hb^{0,\alpha-2}(I_\bhm)$ is uniformly bounded as $\bhm\searrow 0$ for any $N$. Together with the trivial estimate $\|(1-\chi)u\|_{\Hb^{1,\alpha}(I_\bhm)}\lesssim\|u\|_{\Hb^{1,\alpha-1}(I_\bhm)}$, we thus obtain, altogether,
\begin{equation}
\label{EqCImprovedCorner}
  \|u\|_{\Hb^{k,\alpha}(I_\bhm)} \leq C\bigl(\|P_{l,\bhm,\nu}(\omega)u\|_{\Hb^{k-1,\alpha-2}(I_\bhm)} + \|u\|_{\Hb^{2,\alpha-1}(I_\bhm)}\bigr),\qquad
  \alpha\in(\half-l,\tfrac32+l).
\end{equation}
The error term now employs a weak norm both in the sense of derivatives (provided $k\geq 3$) and as $r\sim\bhm\searrow 0$. It is, however, not small as $\bhm\to 0$ when $r\gtrsim 1$, and any further improvements necessarily require the invertibility of the de~Sitter model problem. We are now prepared to prove the main theorem:

\begin{proof}[Proof of Theorem~\usref{ThmIMain}]
  \pfstep{Absence of SdS QNMs near $\omega_0$ which is not a dS QNM.} Suppose first that $\omega\notin\QNM_{\dS,l}(\Lambda,\nu)$. Then Proposition~\ref{PropdSFred} implies the estimate
  \[
    \| u \|_{\Hb^{k,\alpha}([0,2))} \leq C \| P_{l,0,\nu}(\omega)u \|_{\Hb^{k-1,\alpha-2}([0,2))}
  \]
  for fixed $k>\half-\Im\omega$ and $\alpha\in(\half-l,\tfrac32+l)$. The upper bound $2$ in the domain $[0,2)$ can be replaced by any number $R>1$, and the constant $C$ is locally uniform for $R\in(1,\infty)$. We `glue' this estimate into~\eqref{EqCImprovedCorner} by localizing to a neighborhood of the region where $P_{l,\bhm,\nu}(\omega)$ is well approximated by $P_{l,0,\nu}(\omega)$, namely the region $\bhm/r\ll 1$. Thus, fix a cutoff $\chi\in\CIc([0,\tfrac{1}{10}))$, identically $1$ on $[0,\tfrac{1}{20}]$. Fix $\alpha'\in[\alpha-1,\alpha)$ so that $\alpha'\in(\half-l,\tfrac32+l)$ still. Then
  \begin{align*}
    \|\chi(\bhm/r)u\|_{\Hb^{2,\alpha-1}(I_\bhm)} &\lesssim \|\chi(\bhm/r)u\|_{\Hb^{2,\alpha'}([0,2 r_\cC(\bhm)))} \\
      &\lesssim\|P_{l,0,\nu}(\omega)(\chi(\bhm/r)u)\|_{\Hb^{1,\alpha'-2}([0,2 r_\cC(\bhm))} \\
      &\leq \| P_{l,\bhm,\nu}(\omega)(\chi(\bhm/r)u)\|_{\Hb^{1,\alpha'-2}(I_\bhm)} \\
      &\quad\qquad + \|(P_{l,\bhm,\nu}(\omega)-P_{l,0,\nu}(\omega))(\chi(\bhm/r)u)\|_{\Hb^{1,\alpha'-2}(I_\bhm)}
  \end{align*}
  Now, an inspection of~\eqref{EqCMP} shows that
  \[
    (P_{l,\bhm,\nu}(\omega)-P_{l,0,\nu}(\omega))\chi(\bhm/r)(\bhm/r)^{-1} \colon \Hb^{3,\alpha'}(I_\bhm) \to \Hb^{1,\alpha'-2}(I_\bhm)
  \]
  is uniformly bounded; this is the estimate version of the statement that $P_{l,\bhm,\nu}(\omega)$ has $P_{l,0,\nu}(\omega)$ as its model operator at $\bhm/r=0$. Moreover, we have uniform estimates
  \[
    \|(1-\chi(\bhm/r))u\|_{\Hb^{2,\alpha-1}(I_\bhm)},\ \|[P_{l,\bhm,\nu}(\omega),\chi(\bhm/r)]u\|_{\Hb^{1,\alpha'-2}(I_\bhm)} \lesssim \| \tfrac{\bhm}{r}u\|_{\Hb^{2,\alpha'}(I_\bhm)}.
  \]
  Altogether, we thus have
  \[
    \|u\|_{\Hb^{2,\alpha-1}(I_\bhm)} \lesssim \|P_{l,\bhm,\nu}(\omega)u\|_{\Hb^{1,\alpha'-2}(I_\bhm)} + \|\tfrac{\bhm}{r}u\|_{\Hb^{3,\alpha'}(I_\bhm)}.
  \]
  Plugging this into~\eqref{EqCImprovedCorner} and using $\frac{\bhm}{r}\lesssim(\frac{\bhm}{r})^{\alpha-\alpha'}$, we obtain
  \begin{align*}
    \|u\|_{\Hb^{k,\alpha}(I_\bhm)} &\leq C\bigl(\|P_{l,\bhm,\nu}(\omega)u\|_{\Hb^{k-1,\alpha-2}(I_\bhm)} + \|r^{-\alpha'+\alpha}(\tfrac{\bhm}{r})^{\alpha-\alpha'}u\|_{\Hb^{3,\alpha}(I_\bhm)}\bigr) \\
    &= C\bigl(\|P_{l,\bhm,\nu}(\omega)u\|_{\Hb^{k-1,\alpha-2}(I_\bhm)} + \bhm^{\alpha-\alpha'}\|u\|_{\Hb^{3,\alpha}(I_\bhm)}\bigr).
  \end{align*}
  For fixed $k\geq 3$, and for $\bhm\in(0,\bhm_0)$ with $\bhm_0>0$ sufficiently small, the error term is now \emph{small} and can thus be absorbed into the left hand side, giving
  \[
    \|u\|_{\Hb^{k,\alpha}(I_\bhm)} \leq C\|P_{l,\bhm,\nu}(\omega)u\|_{\Hb^{k-1,\alpha-2}(I_\bhm)}.
  \]
  For any \emph{fixed} $\bhm\in(0,\bhm_0)$, this is equivalent to $\|u\|_{H^k(I_\bhm)}\lesssim\|P_{l,\bhm,\nu}(\omega)u\|_{H^{k-1}(I_\bhm)}$, so $P_{l,\bhm,\nu}(\omega)$ has trivial kernel on $H^k(I_\bhm)$. As discussed in~\S\ref{SsCSdS}, this implies that $\omega\notin\QNM_l(\bhm,\Lambda,\nu)$.

  All estimates used here are locally uniform in all parameters; thus, if $\omega_0\notin\QNM_{\dS,l}(\Lambda,\nu_0)$, then there exists $\eps>0$ so that for all $\omega,\nu,\bhm$ with $|\omega-\omega_0|<\eps$, $|\nu-\nu_0|<\eps$, and $\bhm\in(0,\eps)$, one has $\omega\notin\QNM_l(\bhm,\Lambda,\nu)$. This proves (a slight strengthening of) one half of Theorem~\ref{ThmIMain}: any limiting point as $\bhm\searrow 0$ of a sequence of SdS QNMs for a Klein--Gordon field, with mass $\nu$ converging to $\nu_0$, must be a dS QNM for a Klein--Gordon field with mass $\nu_0$.

  \pfstep{Existence of QNMs near a dS QNM.} We next prove the converse, namely that every dS QNM $\omega_0\in\C$ is the limit of a sequence $\omega_\bhm$ of SdS QNMs as $\bhm\searrow 0$. Let
  \[
    K_0=\ker_{\Hb^{k,\alpha}([0,2))}P_{l,0,\nu}(\omega_0),\quad
    K_0^*=\ker_{\Hb^{-k+1,-\alpha-2}([0,2))^\bullet}P_{l,0,\nu}(\omega_0)^*
  \]
  denote the spaces of resonant and dual resonant states, and put $d=\dim K_0=\dim K_0^*$ (thus, $d=1$ by Proposition~\ref{PropdS}, but this is not important). We shall set up a Grushin problem, see \cite[Appendix~C.1]{DyatlovZworskiBook}. We can choose functions
  \[
    u_j^\flat \in \CIc((0,2)),\quad
    u_j^\sharp \in \CIc((0,2)),
  \]
  such that the operators
  \begin{alignat*}{4}
    R_+ &\colon &\Hb^{k,\alpha}([0,2)) \ni&\ u &&\mapsto (\la u,u_j^\flat\ra)_{j=1,\ldots,d} &&\in \C^d, \\
    R_- &\colon &\C^d \ni &\ (w_j)_{j=1,\ldots,d} &&\mapsto \sum_{j=1}^d w_j u_j^\sharp &&\in \CIc((0,2))
  \end{alignat*}
  (with $\la\cdot,\cdot\ra$ denoting the $L^2((0,2);r^2\,\dd r)$ inner product) have the following properties: $R_+$ is injective on $K_0$, and $\ran R_-$ is a complementary subspace to $\ran_{\cX^{k,\alpha}} P_{l,0,\nu}(\omega)\subset\Hb^{k-1,\alpha-2}$. (The latter is equivalent to the nondegeneracy of the matrix $(\la u_j^\sharp,u_k^*\ra)$ where $u_1^*,\ldots,u_d^*$ is a basis of $K_0^*$, which can thus indeed be easily arranged.) Define the operator
  \[
    \wt P_{l,\bhm,\nu}(\omega) = \begin{pmatrix} P_{l,\bhm,\nu}(\omega) & R_- \\ R_+ & 0 \end{pmatrix} \colon \cX_\bhm^{k,\alpha} \oplus \C^d \to \Hb^{k-1,\alpha-2}(I_\bhm) \oplus \C^d.
  \]
  By construction, its model problem at $\bhm=0$, $r>0$ is invertible for $\omega=\omega_0$, and the Schwarzschild model problem is unaffected by the finite-dimensional modification since $u_j^\flat,u_j^\sharp$ are supported away from $r=0$. Thus, a (notational) adaptation of the previous arguments, culminating in the estimate
  \begin{equation}
  \label{EqCtildeP}
    \|(u,v)\|_{\Hb^{k,\alpha}(I_\bhm)\oplus\C^d} \leq C\|\wt P_{l,\bhm,\nu}(\omega)(u,v)\|_{\Hb^{k-1,\alpha-2}(I_\bhm)\oplus\C^d},
  \end{equation}
  implies that $\wt P_{l,\bhm,\nu}(\omega)$ is invertible for $|\omega-\omega_0|<\eps$, $\bhm\in(0,\eps)$, for $\eps>0$ sufficiently small; write
  \[
    \wt P_{l,\bhm,\nu}(\omega)^{-1} = \begin{pmatrix} A_\bhm(\omega) & B_\bhm(\omega) \\ C_\bhm(\omega) & D_\bhm(\omega) \end{pmatrix}.
  \]
  By the Schur complement formula then, the invertibility of $P_{l,\bhm,\nu}(\omega)$ is equivalent to that of the $d\times d$ matrix $D_\bhm(\omega)$. We claim that $D_\bhm(\omega)$ depends continuously on $\bhm\in[0,\eps)$, with values in $d\times d$ matrices whose entries are holomorphic for $|\omega-\omega_0|<\eps$. Granted this, then since $\det D_\bhm(\omega)$ has an isolated zero at $(\bhm,\omega)=(0,\omega_0)$, Rouch\'e's theorem implies the existence of $\omega_\bhm$, depending continuously on $\bhm$, with $\lim_{\bhm\searrow 0}\omega_\bhm\to\omega_0$ and $\det D_\bhm(\omega_\bhm)=0$, so $\omega_\bhm\in\QNM_l(\bhm,\Lambda,\nu)$.

  It remains to prove that $D_\bhm(\omega)\to D_0(\omega)$ uniformly for $|\omega-\omega_0|\leq\eps'<\eps$; since $D_0$ is analytic (hence continuous) in $\omega$, it is enough to establish pointwise convergence. Thus, let $w\in\C^d$ and put
  \[
    \bigl(u_\bhm(\omega),v_\bhm(\omega)\bigr) := \wt P_{l,\bhm,\nu}(\omega)^{-1}(0,w);
  \]
  we need to show that
  \begin{equation}
  \label{EqCSubseqConv}
    v_\bhm(\omega)=D_\bhm(\omega)w\to v_0(\omega).
  \end{equation}
  By~\eqref{EqCtildeP}, $u_\bhm(\omega)$ is uniformly bounded in $\Hb^{k,\alpha}(I_\bhm)$, and $v_\bhm(\omega)$ is uniformly bounded in $\C^d$. Extend $u_\bhm(\omega)$ to $\tilde u_\bhm(\omega)\in\Hb^{k,\alpha}([0,3))$ with $\|\tilde u_\bhm(\omega)\|_{\Hb^{k,\alpha}([0,3))}\leq C\|u_\bhm(\omega)\|_{\Hb^{k,\alpha}(I_\bhm)}$. Consider any weak (thus distributional) subsequential limit $\tilde u_0(\omega)\in\Hb^{k,\alpha}([0,3))$ of $\tilde u_{\bhm_j}(\omega)$ where $\bhm_j\to 0$, and select the subsequence so that also $v_\bhm(\omega)\to v_0$; then
  \[
    \tilde P_{l,0,\nu}(\omega)\bigl(\tilde u_0(\omega)|_{[0,2)},v_0\bigr)=(0,w).
  \]
  But since $\tilde P_{l,0,\nu}(\omega)$ is invertible, this means that necessarily $v_0=\pi_2(\wt P_{l,0,\nu}(\omega)^{-1}(0,w))$, where $\pi_2$ denotes projection to the second summand. This establishes~\eqref{EqCSubseqConv}.

  \pfstep{Smooth convergence of mode solutions.} We now use $d=1$ and denote by $\omega_\bhm\in\QNM_l(\bhm,\Lambda,\nu)$ the QNM tending to $\omega_0$ as $\bhm\searrow 0$. Note then that $P_{l,\bhm,\nu}(\omega_\bhm)u_\bhm=0$ if and only if $\wt P_{l,\bhm,\nu}(\omega_\bhm)(u_\bhm,0)=(0,R_+ u_\bhm)$; since $\wt P_{l,\bhm,\nu}(\omega_\bhm)$ is injective for small $\bhm$, we can thus \emph{construct} $u_\bhm$ as
  \[
    u_\bhm = \pi_1\bigl(\wt P_{l,\bhm,\nu}(\omega_\bhm)^{-1}(0,1)\bigr),
  \]
  where $\pi_1$ denotes projection to the first summand. The discussion after~\eqref{EqCSubseqConv} implies that for any fixed compact interval $I\Subset(0,2]$, we have weak convergence $u_\bhm\weakto u_0$ in $H^k(I)$, thus strong convergence in $H^{k-\eps}(I)$ and therefore in $\cC^{k-1}(I)$ by Sobolev embedding. Since $k$ is arbitrary, this proves the smooth convergence of $u_\bhm|_I$ to $u_0|_I$. The proof of Theorem~\ref{ThmIMain} is complete.
\end{proof}

\section{Numerics}
\label{SN}

We wish to find quasinormal modes and mode solutions for the Klein-Gordon operator $\Box_{g_\bhm}-\nu$ on a Schwarzschild--de~Sitter spacetime numerically. Restricting to time frequency $\omega\in\C$ and decomposing into spherical harmonics, we obtain an ODE of the form
\begin{equation}
\label{EqNOp}
    a_0(x,\omega)u(x)+a_1(x,\omega)u'(x)+a_2(x,\omega)u''(x)+\ldots+a_n(x,\omega)u^{(n)}(x)=0,
\end{equation}
where the coefficients $a_k(x,\omega)$ are polynomials in $\omega$. (See equation~\eqref{EqCMP}, which is indeed of this form with $n=2$ and $x=r$.) For now, we shall work with the general form~\eqref{EqNOp}. We seek to find values of $\omega$ for which the equation~\eqref{EqNOp} has smooth solutions over some given interval $[a,b]$.

We shall modify Jansen's approach \cite{JansenMathematica} so as to obtain higher precision for small black holes. We also extend \cite{JansenMathematica} to numerically calculate the dual resonant states.

\subsection{Description of the method}
\label{SsNM}

Following Jansen's approach, we select $N+1$ grid points $x_0,\ldots,x_N$ on the interval $[a,b]$, where $N\in\N$ is the degree of our approximation. We define a basis of Lagrange polynomials
\[
  C_k(x)=\prod_{j\neq k}\frac{x-x_j}{x_k-x_j}.
\]
Given any continuous function $f(x)$, its Lagrange interpolant is $f_N(x)=\sum_{k=0}^N f(x_k) C_k(x)$; we can view the coefficients $f(x_k)$ as a vector $\mathbf{f}_N$. We discretize multiplication by a function $a(x)$ as the diagonal matrix with entries $a(x_k)$; differentiation in $x$ is discretized as the matrix with $(j,k)$ entry $C'_k(x_j)$. The discretized differential operator can now be expressed as a matrix
\begin{equation}
\label{DiscreteOp}
    \mathbf{O}_N(\omega):=\mathbf{M}_0+\omega\mathbf{M}_1+\omega^2\mathbf{M}_2+\ldots+\omega^p\mathbf{M}_p,
\end{equation}
where $p\in\N$ is the maximal polynomial degree of the coefficients $a_k(x,\omega)$ in $\omega$. QNMs are found numerically by finding those values of $\omega$ for which $\mathbf{O}_N(\omega)$ is singular. A nonzero vector in the null space, which we will call $\mathbf{u}_N=(u_0,\ldots,u_n)$, gives an approximation $\sum_k u_k C_k(x)$ to a resonant state. 

The task of finding such $\omega\in\C$ is conveniently expressed as a generalized eigenvalue problem: defining the matrices
\[
    \mathbf{A}=
    \begin{pmatrix}
        \mathbf{M}_0 & \mathbf{M}_1 & \mathbf{M}_2 & \ldots & \mathbf{M}_{p-1}\\
        0 & \mathbbm{1} & 0 & \ldots & 0\\
        0 & 0 & \mathbbm{1} & \ldots & 0\\
        \vdots & \vdots & \vdots & \ddots & \vdots\\
        0 & 0 & 0 & \ldots & \mathbbm{1}
    \end{pmatrix},
    \qquad
    \mathbf{B}=
    \begin{pmatrix}
        0 & 0 & 0 & \ldots & 0 &\mathbf{M}_p\\
        -\mathbbm{1} & 0 & 0 & \ldots & 0 & 0\\
        0 & -\mathbbm{1} & 0 & \ldots & 0 & 0\\
        \vdots & \vdots & \vdots & \ddots & \vdots & \vdots\\
        0 & 0 & 0 & \ldots & -\mathbbm{1} & 0
    \end{pmatrix},
\]
solving $\det\mathbf{O}_N(\omega)=0$ is equivalent to find the eigenvalues $\omega$ of the generalized eigenvalue problem $\mathbf{A}\mathbf{u}=\omega\mathbf{B}\mathbf{u}$.

\subsubsection{Extension to find dual states}
\label{dualStateProcedure}

In order to find dual states, we compute the nullspace of the conjugate transpose $\mathbf{O}_N(\omega)^*$ instead; denote a nonzero element of $\ker\mathbf{O}_N(\omega)^*$ by $\mathbf{v}_N$. We want to interpret $\mathbf{v}_N$ as a distribution, or rather an approximation inside the finite-dimensional space $\cF_N:=\mathspan\{C_0,\ldots,C_N\}$, where we identify $\cF_N$ with its dual space by choosing an $L^2$ inner product $\la-,-\ra$ for functions on $[a,b]$. A dual basis $C_0^*,\ldots,C_N^*\in\cF_N$ is then defined by the requirement
\[
  \la C_i^*,C_j\ra = \delta_{i j}.
\]
If we write this dual basis in terms of the basis functions as $C^*_i=\sum_k T_{k i}C_k$, we have
\[
  \sum_k T_{k i} \braket{C_k,C_j}=\delta_{i j}
\]
Defining $U_{k j}=\braket{C_k,C_j}$, and $\mathbf{T}=(T_{m m'})$, $\mathbf{U}=(U_{m m'})$, this means
\[
  \mathbf{T}^T \mathbf{U}=\mathbf{I}\qquad\implies\qquad\mathbf{T}=(\mathbf{U}^T)^{-1}.
\]
Thus, we can use the matrix $\mathbf{U}$ to convert the coefficients $\mathbf{v}_N$ with respect to the dual basis into coefficients for the original basis. This gives
\[
  \mathbf{u}^*_N=(\mathbf{U}^T)^{-1}\mathbf{v}_N.
\]
Writing $\mathbf{u}^*_N=(u_0^*,\ldots,u_N^*)$, the function $\sum_k u_k^* C_k(x)$ is then our approximation to the dual state.

When studying the convergence of dual states on SdS in the small mass limit $\bhm\searrow 0$, one must keep in mind that the convergence necessarily takes place in a rather weak topology; we shall test for convergence by integrating against polynomial test functions. Given a test function $\phi$, discretized as the vector $(\phi_k)=(\phi(x_k))$, the $L^2$-inner product of the elements of $\cF_N$ with coefficients $\mathbf{u}_N^*=(\mathbf{U}^T)^{-1}\mathbf{v}_N$ (where $\mathbf{v}_N=(v_0,\ldots,v_N)$) and $(\phi_k)$ is
\[
  \sum_{i j k} \la (\mathbf{U}^{-1})_{j i}v_j  C_i, \phi_k C_k \ra = \sum_{i j k} (\mathbf{U}^{-1})_{j i}v_j\ol{\phi_k}\mathbf{U}_{i k} = \sum_k v_k \ol{\phi_k}.
\]
Thus, we do not actually need to compute the matrix $\mathbf{U}$ and its inverse in order to perform our convergence tests for dual states, which is numerically advantageous.

\subsection{Choice of coordinates}

The outgoing boundary conditions \eqref{EqIOutgoing} can be absorbed into the definition of our time coordinate if we choose $t'_\bhm=t+h_\bhm(r)$ where
\[
  h_\bhm=\frac{1}{2\kappa_+(\bhm)}\log(r-r_+(\bhm))+\frac{1}{2\kappa_\cC(\bhm)}\log(r-r_\cC(\bhm)).
\]
In these coordinates, curves of constant time are smooth across the event and cosmological horizons. (Moreover, this explicit definition is more convenient for numerical computations than the definition~\eqref{EqCMTimeFn}, \eqref{EqCMTimeFnMod} using smooth cutoff functions.) Thus, the boundary condition for mode solutions becomes smoothness across the horizons. We thus consider the action of $\Box_{g_\bhm}-\mu$ on separated functions $e^{-i\omega t'_\bhm}u(r)Y_{l m}(\theta,\phi)$, and the QNMs and mode solutions can be found using the numerical method described above.

As illustrated in Figure~\ref{FigIBlowup}, there are two limiting regimes as we decrease the black hole mass. One corresponds to de Sitter space and occurs at scales of $r\sim 1$ and one corresponding to Schwarzschild space at scales of $r\sim\bhm$, or $\hat{r}\sim 1$. Hence for small $\bhm$, the latter regime is poorly resolved if we use an equally spaced or Chebyshev grid on $[r_+(\bhm),r_\cC(\bhm)]$. To resolve this limit more accurately, we want to increase the number of grid points near the event horizon so that the spacing between them is less of order $\bhm$.

Arbitrary choices of grid points are subject to Runge's phenomenon and lead to poor convergence when the number of grid points is increased. To achieve good resolution of both limiting regimes while also having good convergence, we transform our radial coordinate to $x=\log r$. The event and cosmological horizons now become $x'_+(\bhm)\approx-\log(2\bhm)$ and $x'_\cC(\bhm)\approx 0$, thus we stretch out the region close to the event horizon. (Noting that $\pa_x=r\pa_r$, this is compatible also with the theoretical considerations in~\S\ref{SsCPf}.). We can then use Chebyshev nodes on the interval $[x'_+(\bhm),x'_\cC(\bhm)]$ for the Klein-Gordon equation in these coordinates to mitigate Runge's phenomenon and to resolve modes for small black hole masses.

\subsection{Results and conjectures}
\label{SsNR}

In Figure~\ref{FigIQNM}, we plot the QNMs for a SdS black hole for the massless scalar field for $\Lambda=3$ and $\bhm=0.150,0.075,0.000$. From this figure, we see QNMs along the imaginary axis which converge to QNMs of the limiting pure de Sitter spacetime. We also see QNMs away from the imaginary axis which move away from the origin and conjecturally are approximately equal to $\bhm^{-1}$ times modes of the mass $1$ Schwarzschild spacetime (see the discussion around~\eqref{EqIOmegaResc}). Further, tracking individual QNMs as $\bhm\searrow 0$, we observe the convergence of the mode solutions to the de Sitter mode solutions. This confirms Theorem~\ref{ThmIMain}.

\begin{prob}[Convergence of QNMs in a half space]
\label{PropNRConv}
  Show that Theorem~\usref{ThmIMain} holds without restriction to fixed angular momenta; prove convergence of QNMs in any fixed half space $\Im\omega>-C$.
\end{prob}

\begin{figure}[!ht]
    \centering
    \subfloat[$l=0$]{\includegraphics[width=0.5\textwidth]{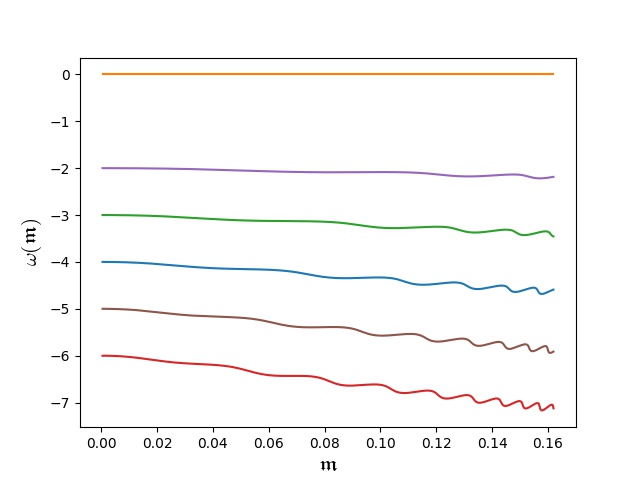}}
     \subfloat[$l=1$]{\includegraphics[width=0.5\textwidth]{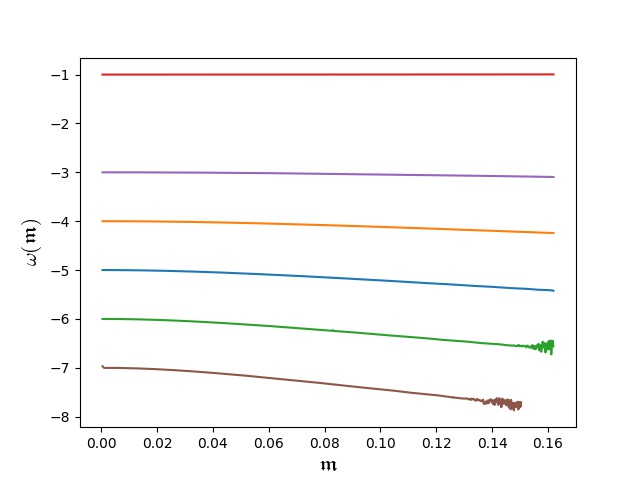}}
    \hfill
    \subfloat[$l=0$]{\includegraphics[width=0.5\textwidth]{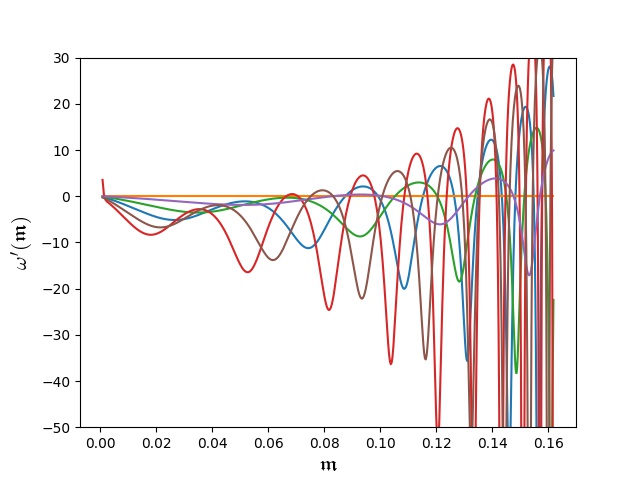}}
    \subfloat[$l=1$]{\includegraphics[width=0.5\textwidth]{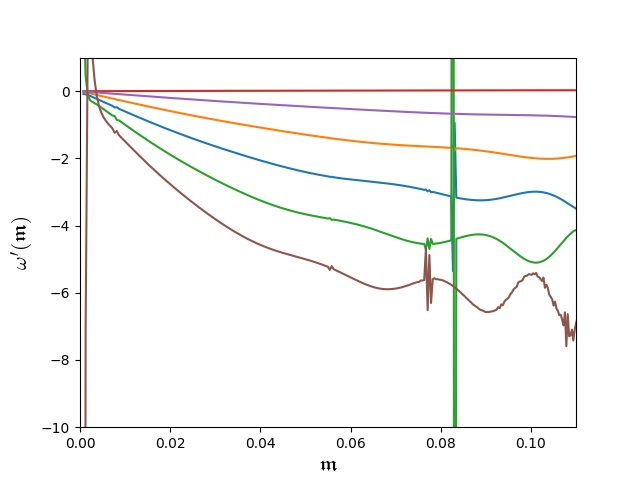}}
    \caption{Imaginary part of massless scalar field QNMs in SdS spacetime ($\Lambda=3$, $\nu=0$) with angular momentum (A) $l=0$ and (B) $l=1$ as a function of black hole mass. Derivative of QNMs of a massless scalar field in SdS spacetime with angular momentum (C) $l=0$ and (D) $l=1$ as a function of black hole mass.}
    \label{SdSImModes}
\end{figure}

Figure~\ref{SdSImModes} shows the imaginary parts of convergent (as $\bhm\to 0$) QNMs $\omega(\bhm)$ for the massless scalar field ($\nu=0$) as a function of $\bhm$, with $\Lambda=3$. We see convergence to the values of the QNMs in dS space. We also compute the numerical derivative $\omega'(\bhm)$ of the QNMs. This derivative appears to vanish for most modes as $\bhm\rightarrow 0$; we believe that the derivatives which do not vanish are numerical artifacts. This leads us to conjecture that to first order in $\bhm$, QNMs of massless scalar fields for small black hole mass SdS spacetime are the same as those in dS spacetime. In \S\ref{perturbationCalc} we perform a naive perturbation theory calculation that supports this conjecture.

\begin{figure}[!ht]
    \centering
    \subfloat[]{\includegraphics[width=0.5\textwidth]{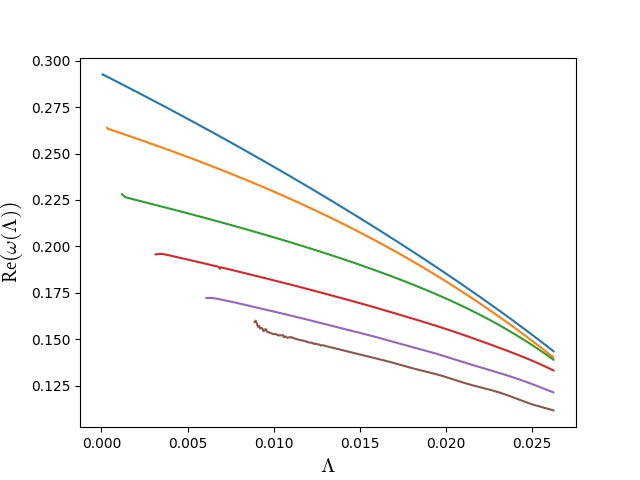}}
    \subfloat[]{\includegraphics[width=0.5\textwidth]{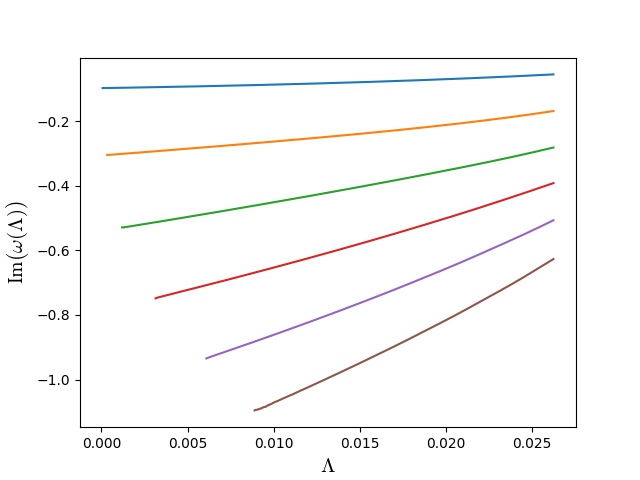}}
    \caption{QNMs of a massless scalar field ($\nu=0$) on SdS with $\bhm=1$ as a function of $\Lambda$. (A) Real part of modes converging to the Schwarzschild limit as $\Lambda\searrow 0$. Only modes with positive real part are shown, those with negative real part are mirrored. (B) Imaginary part of modes converging to the Schwarzschild limit.}
    \label{SConvModes}
\end{figure}

In Figure~\ref{SConvModes}, we again consider massless scalar fields, but now fixing $\bhm=1$ and letting $\Lambda$ vary. (This is equivalent, up to scaling of QNMs by $\bhm$, to taking $\bhm\searrow 0$ and keeping $\Lambda$ fixed, cf.\ \eqref{EqIN}.) We plot the real and imaginary parts of SdS QNMs that converge to a nonzero limit; the limits are, conjecturally, QNMs of a mass $1$ Schwarzschild black hole. (Since our calculations are done at fixed $\Lambda$ but with $\bhm$ tending to zero, in which the QNMs of interest scale like $\bhm^{-1}$ and thus become very large, we have poor resolution for $\bhm=1$ at small values of $\Lambda$. Hence, we cannot track the QNMs accurately down to $\Lambda=0$ here.)

\begin{figure}[!ht]
    \centering
    \subfloat[]{\includegraphics[width=0.5\textwidth]{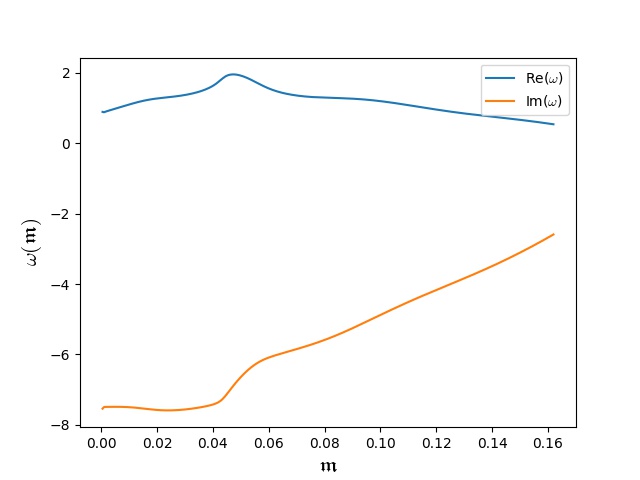}}
    \subfloat[]{\includegraphics[width=0.5\textwidth]{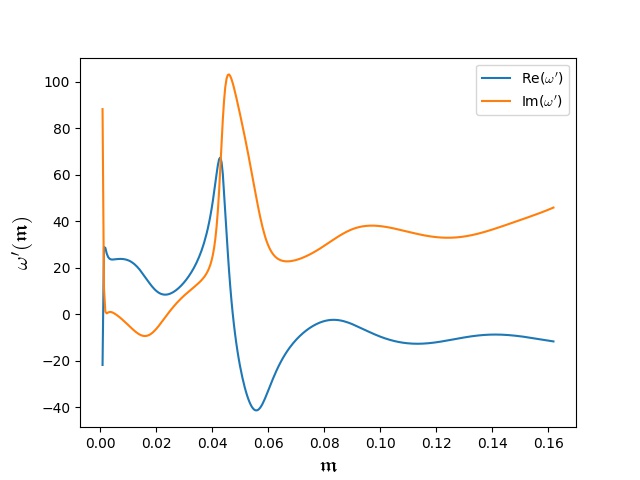}}
    \caption{(a) Convergence of a mode of a massive scalar field ($\Lambda=3$, $l=0$, $\nu=3$) as a function of $\bhm$. (b) Derivative of the same mode with respect to $\bhm$.}
    \label{MassiveModes}
\end{figure}

\begin{figure}[!ht]
    \centering
    \subfloat[$l=0$]{\includegraphics[width=0.5\textwidth]{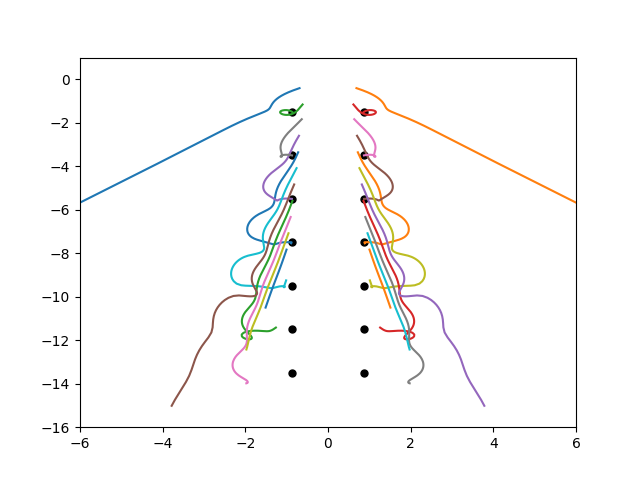}}
    \subfloat[$l=1$]{\includegraphics[width=0.5\textwidth]{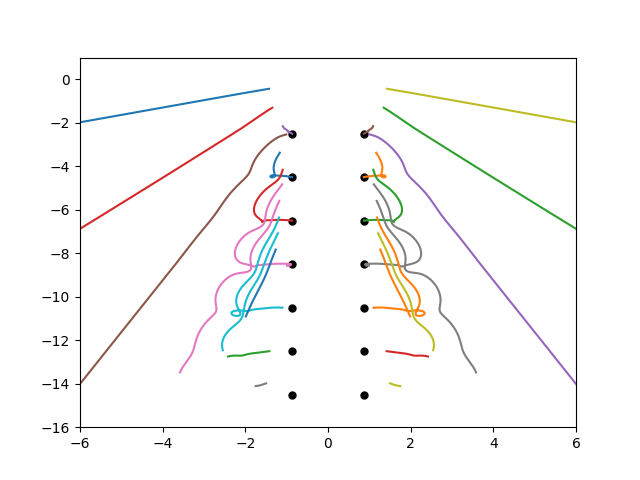}}
    \caption{Trails followed by the QNMs of a massive scalar field ($\Lambda=3$, $\nu=3$) for (A) $l=0$ and (B) $l=1$. The QNMs in the de Sitter limit are shown in black.}
    \label{SdSTrailsMassive}
\end{figure}

In Figure~\ref{MassiveModes}, we turn to massive scalar fields. We fix $l=0$, $\nu=3$, and plot the real and imaginary part of a QNM $\omega(\bhm)\in\QNM_l(3,\bhm,\nu)$ on SdS converging to a dS QNM as a function of $\bhm$ while fixing $\Lambda=3$. We also show the derivative of $\omega(\bhm)$ with respect to the black hole mass. Contrary to the case of the massless scalar field, the numerics indicate that the dependence $\omega(\bhm)$ on $\bhm$ is nontrivial to first order, i.e.\ $\omega'(0)\neq 0$. In Figure~\ref{SdSTrailsMassive}, we plot the trails followed by large mass ($\nu=3$) scalar field QNMs on SdS ($\Lambda=3$) as we vary $\bhm$. We clearly see modes which appear to diverge to infinity, conjecturally corresponding to modes converging to the Schwarzschild limit, along with modes which converge to the values in the dS limit.

\begin{prob}[Appearance of Schwarzschild QNMs]
  Fix $C>0$. Show that the set $\QNM(\Lambda,\bhm,\nu)\cap\{\Im\omega>\max(-C\bhm^{-1},-C\Re\omega)\}$ is close to $\bhm^{-1}$ times the QNM spectrum of a mass $1$ Schwarzschild black hole.
\end{prob}

\begin{rmk}[Larger subsets of $\C$]
\label{RmkNR}
  A description of $\QNM(\Lambda,\bhm,\nu)$ in $\Im\omega\geq-C\bhm^{-\alpha}$ even for small $\alpha>0$ is likely considerably more subtle due to the accumulation of dS QNMs along the negative imaginary axis. This is related to the fact that the Schwarzschild resolvent does not extend meromorphically to the complex plane, but at best to the logarithmic cover of $\C\setminus i(-\infty,0]$, cf.\ \cite{SaBarretoZworskiResonances,DonningerSchlagSofferSchwarzschild,TataruDecayAsympFlat,HintzPrice}.
\end{rmk}

\subsubsection{Dual States}

As explained in \S\ref{dualStateProcedure}, we extended the procedure \cite{JansenMathematica} to also compute dual states. Plots of these numerical states are shown in Figure \ref{duals}. However, these are \emph{distributions} rather than smooth functions; in the case of pure dS state for the massless field, one can show almost all dual states are sums of differentiated $\delta$-distributions supported on the cosmological horizon, as discussed in \cite{HintzXiedS}. Thus, we only check for weak convergence of our numerical dual states by integrating against various test functions. The results are shown in Tables \ref{l0dual} and \ref{l1dual}.

\begin{figure}[!ht]
    \centering
    \subfloat[]{\includegraphics[width=0.5\textwidth]{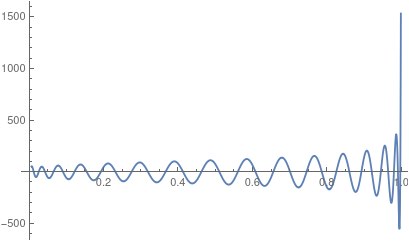}}
    \subfloat[]{\includegraphics[width=0.5\textwidth]{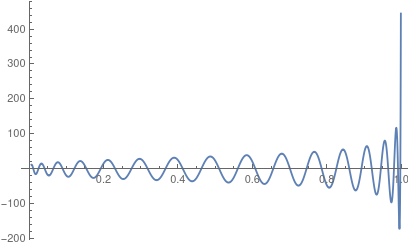}}
    \caption{Fix $\Lambda=3$, $\nu=0$. (A) Numerical dual resonant state that converges to the $l=0,\omega=-2i$ dual resonant state ($-2\delta(r-1)-\delta'(r-1)$) in dS space. (B) Numerical dual resonant state that converges to the $l=1,\omega=-3i$ dual resonant state ($2\delta'(r-1)+\delta''(r-1)$) in dS space.}
    \label{duals}
\end{figure}

\begin{table}[!ht]
	\centering
	\begin{tabular}{|c|c|c|c|c|c|}
		\hline\rowcolor[gray]{0.9}[\tabcolsep]
		$\bhm$ & 1 & $r-r_\cC$ & $(r-r_\cC)^2/2!$ & $(r-r_\cC)^3/3!$ & $(r-r_\cC)^4/4!$\\\hline
		0.1000 & -0.786 & 1.0 & -0.338 & 0.0331 & -0.004\\\hline
		0.0500 & 4.2237 & 1.0 & -1.747 & 0.3052 & -0.049\\\hline
		0.0250 & -3.160 & 1.0 & 0.3107 & -0.060 & 0.0107\\\hline
		0.0100 & -2.271 & 1.0 & 0.0706 & -0.014 & 0.0026\\\hline
		0.0050 & -2.120 & 1.0 & 0.0311 & -0.006 & 0.0012\\\hline
		0.0025 & -2.056 & 1.0 & 0.0146 & -0.003 & 0.0005\\\hline
		0 & -2 & 1 & 0 & 0 & 0\\\hline
	\end{tabular}
	\caption{Fix $\Lambda=3$, $\nu=0$. Stated are integrals (with respect to the integration measure $\dd r$) of the numerical dual state at $l=0,\omega\approx-2i$ (normalized so that all entries in the third column are $1$) against various test functions for various black hole masses; here $r_\cC=r_\cC(\bhm)$. The exact value is given for pure dS ($\bhm=0$) using the results in \cite{HintzXiedS}: the dS dual state is $-2\delta(r-1)-\delta'(r-1)$. We observe convergence to the values in pure dS space.}
	\label{l0dual}
\end{table}

\begin{table}[!ht]
	\centering
	\begin{tabular}{|c|c|c|c|c|c|}
		\hline\rowcolor[gray]{0.9}[\tabcolsep]
		$\bhm$ & 1 & $r-r_\cC$ & $(r-r_\cC)^2/2!$ & $(r-r_\cC)^3/3!$ & $(r-r_\cC)^4/4!$\\\hline
		0.1000 & 0.4495 & -1.339 & 1.0 & -0.206 & -0.004\\\hline
		0.0500 & 6.0294 & -7.200 & 1.0 & 0.5644 & -0.050\\\hline
		0.0250 & 1.3234 & -3.076 & 1.0 & 0.0935 & -0.009\\\hline
		0.0100 & 0.4364 & -2.352 & 1.0 & 0.0273 & -0.003\\\hline
		0.0050 & 0.2085 & -2.168 & 1.0 & 0.0125 & -0.001\\\hline
		0.0025 & 0.1013 & -2.078 & 1.0 & 0.0042 & -9.952\\\hline
		0 & 0 & -2 & 1 & 0 & 0\\\hline
	\end{tabular}
	\caption{Integrals (with respect to $\dd r$) of the radial part of the numerical dual state at $l=1,\omega=-3i$ (normalized so that all entries in the fourth column are $1$) against various test functions for various black hole masses; here $r_\cC=r_\cC(\bhm)$. The exact value is given for pure dS ($\bhm=0$) using the results in \cite{HintzXiedS}: the radial dependence of the dS dual state is $2\delta'(r-1)+\delta''(r-1)$. We observe convergence to the values in pure dS space. The penultimate entry in the last column is likely a numerical artefact.}
	\label{l1dual}
\end{table}

\subsection{Naive perturbation theory calculation}
\label{perturbationCalc}

This section is heuristic, i.e.\ we do not carefully state function spaces etc. Suppose that the frequency of a specific mode $\omega(\bhm)$ has an expansion for small masses
\[
  \omega(\bhm)=\omega(0)+\bhm\omega'(0)+\cO(\bhm^2).
\]
Suppose also that the corresponding mode solution can be expanded as
\[
  u_\bhm(r,\theta,\phi) = u_0(r,\theta,\phi)+\bhm\left.\partial_\bhm u_\bhm(r,\theta,\phi)\right|_{\bhm=0}+\cO(\bhm^2).
\]
Note that the (formal) expansion of the spectral family in \eqref{EqISdSSpec} around $\bhm=0$ reads
\[
  P_{\bhm,\Lambda,\nu}(\omega(\bhm))=P_{0,\Lambda,\nu}(\omega(0))+\bhm\left.\partial_\bhm P_{\bhm,\Lambda,\nu}(\omega(0))\right|_{\bhm=0}+\bhm \omega'(0)P'_{0,\Lambda,\nu}(\omega(0))+\cO(\bhm^2).
\]
Then to first order we have
\begin{align*}
    0 &=P_{\bhm,\Lambda,\nu}(\omega(\bhm))u_\bhm(r,\theta,\phi) \\
    & =  P_{0,\Lambda,\nu}(\omega(0))u_0(r,\theta,\phi) \\
    &\qquad + \bhm\Bigl(\partial_\bhm P_{\bhm,\Lambda,\nu}(\omega(0))|_{\bhm=0}u_0(r,\theta,\phi) + \omega'(0)P'_{0,\Lambda,\nu}(\omega(0))u_0(r,\theta,\phi) \\
    &\qquad\hspace{15em} + P_{0,\Lambda,\nu}(\omega(0))\left.\partial_\bhm u_\bhm(r,\theta,\phi)\right|_{\bhm=0}\Bigr) + \cO(\bhm^2).
\end{align*}
The zeroth order term vanishes. Then taking the inner product with the de~Sitter dual state $u_0^*$ gives
\begin{equation}
    \omega'(0)=-\frac{\braket{\left.\partial_\bhm P_{\bhm,\Lambda,\nu}(\omega(0))\right|_{\bhm=0}u_0,u^*_0}}{\braket{P'_{0,\Lambda,\nu}(\omega(0))u_0,u^*_0}}.
\end{equation}
(The term involving $\pa_\bhm u_\bhm|_{\bhm=0}$ vanishes upon integration by parts.) In suitable coordinates, and defining $\alpha_\bhm$ as in \eqref{EqCMalpha}, $P_{\bhm,\Lambda,\nu}$ takes the form given by \eqref{EqCMP}. We can then substitute for $u_0$ and $u_0^*$ the explicit expressions found in \cite{HintzXiedS} to find that $\omega'(0)=0$; that is, to first order there is no correction to the QNM frequencies for a massless scalar field. (For example, for the $l=1$, $\omega(0)=-i$ state, we have $u_0=r$ and $u_0^*=\delta(r-1)$.)

\begin{prob}[Asymptotic expansion]
\label{ProbNRPert}
  Justify the above perturbative calculation for $\nu=0$, and prove a detailed asymptotic expansion of $\omega_\bhm$ as $\bhm\searrow 0$ for general $\nu$.
\end{prob}

\bibliographystyle{alpha}
\newcommand{\etalchar}[1]{$^{#1}$}

\end{document}